\documentclass[11pt]{article}

\newif\ifmaincontent
\newif\ifsupplementcontent
\ifdefined\supplementaryonly
    \maincontentfalse
    \supplementcontenttrue
\else
    \maincontenttrue
    \ifdefined\submissionmainonly
        \supplementcontentfalse
    \else
        \supplementcontenttrue
    \fi
\fi

\usepackage[letterpaper,margin=1in]{geometry}
\usepackage[T1]{fontenc}
\usepackage[utf8]{inputenc}
\usepackage{newtxtext}
\usepackage{microtype}
\usepackage{setspace}
\usepackage{graphicx}
\graphicspath{{figures/}}
\usepackage{caption}
\usepackage{placeins}
\usepackage{booktabs}
\usepackage{array}
\usepackage{tabularx}
\usepackage{xcolor}
\usepackage[hidelinks]{hyperref}
\ifdefined\supplementaryonly
    \usepackage{xr-hyper}
    \hypersetup{
        pdftitle={Supplementary Information for Learned-projector QAOA for hierarchical optimization},
        pdfauthor={Kangyun Zhou}
    }
\else
    \ifdefined\submissionmainonly
        \usepackage{xr-hyper}
    \fi
    \hypersetup{
        pdftitle={Learned-projector QAOA for hierarchical optimization},
        pdfauthor={Kangyun Zhou}
    }
\fi

\usepackage{amsmath,amssymb,amsthm,mathtools}
\usepackage{enumitem}
\numberwithin{equation}{section}

\theoremstyle{definition}

\newtheorem{assumption}{Assumption}[section]
\theoremstyle{plain}
\newtheorem{theorem}{Theorem}[section]
\newtheorem{lemma}{Lemma}[section]

\theoremstyle{remark}
\newtheorem{remark}{Remark}[section]

\usepackage[backend=biber,style=numeric,sorting=none,maxbibnames=99]{biblatex}
\AtBeginBibliography{\small}
\newcommand{\bits}{\{0,1\}}
\newcommand{\ket}[1]{\left|#1\right\rangle}
\newcommand{\bra}[1]{\left\langle#1\right|}
\newcommand{\braket}[2]{\left\langle#1\middle|#2\right\rangle}

\newcommand{\norm}[1]{\left\lVert#1\right\rVert}
\newcommand{\Id}{\mathbb{I}}
\newcommand{\RU}{\mathrm{RU}}

\newcommand{\argmin}{\operatorname*{arg\,min}}

\newcommand{\R}{\mathbb{R}}

\newcommand{\calG}{\mathcal{G}}
\newcommand{\calH}{\mathcal{H}}

\newcommand{\calO}{\mathcal{O}}

\newcommand{\suppref}[1]{%
    \ifdefined\submissionmainonly
        Supplementary Section~\ref{#1}%
    \else
        \ifdefined\supplementaryonly
            Section~\ref{#1}%
        \else
            Appendix~\ref{#1}%
        \fi
    \fi
}

\ifdefined\supplementaryonly
    \title{\textit{Learned-projector QAOA for hierarchical optimization}\\[0.6em]{\large Supplementary Information}}
\else
    \title{Learned-projector QAOA for hierarchical optimization}
\fi
\author{Kangyun Zhou$^{1,2}$, Dong An$^{3}$, Jin-Peng Liu$^{1,4,5,\thanks{Corresponding author: liujinpeng@tsinghua.edu.cn}}$\\
\footnotesize $^{1}$ Yau Mathematical Sciences Center, Tsinghua University\\
\footnotesize $^{2}$ Qiuzhen College, Tsinghua University\\
\footnotesize $^{3}$ Beijing International Center for Mathematical Research, Peking University\\
\footnotesize $^{4}$ Institute for Applied Mathematics, Tsinghua University\\
\footnotesize $^{5}$ Beijing Institute of Mathematical Sciences and Applications}

\ifdefined\supplementaryonly
    \date{}
\else
    \date{}
\fi

\begin{document}
\ifmaincontent
\maketitle

\begin{abstract}
Many optimization problems reveal inexpensive structural information before requiring costly final evaluation, whereas the conventional quantum approximate optimization algorithm (QAOA) applies a single aggregate objective throughout. We introduce the learned-projector quantum alternating operator ansatz (LP-QAOA), a multistage protocol that freezes optimized circuits and uses their output states to define later projector mixers. We prove a stability theorem for LP-QAOA that bounds the propagation of approximation errors through successive frozen stages. Our analysis also shows how learned-projector mixing avoids the high-order tunnelling suppression of local mixers. We conduct state-vector simulations on block-constrained spin tiling (BCST) instances, where LP-QAOA achieves higher probabilities of sampling the optimum, lower estimated logical-resource requirements, and better trainability than the tested baselines. A supporting stochastic block model experiment extends LP-QAOA to optimization problems with soft hierarchical structure. Our results highlight the potential of LP-QAOA to improve variational quantum optimization by exploiting hierarchical problem structure.
\end{abstract}

\section{Introduction}
\label{sec:introduction}

Quantum algorithms approach combinatorial optimization by encoding objectives and constraints in Hamiltonians and using quantum dynamics to concentrate probability on high-quality solutions. Quantum annealing and adiabatic quantum computation (AQC) were among the earliest general approaches \cite{kadowaki1998quantum,farhi2000quantum}. In AQC, a simple driver Hamiltonian is slowly transformed into a problem Hamiltonian whose ground state represents an optimum. This construction inspired the quantum approximate optimization algorithm (QAOA), which alternates evolutions under cost and mixer Hamiltonians and optimizes their durations classically \cite{farhi2014quantum}. Product-formula discretization of a sufficiently slow interpolation supplies convergent QAOA schedules as depth $p$ increases, while variational optimization can identify finite-time nonadiabatic schedules \cite{farhi2014quantum,binkowski2024elementary,zhou2020quantum,brady2021optimal}. Its tunable depth and hybrid quantum--classical structure make QAOA a prominent candidate for noisy intermediate-scale quantum processors \cite{farhi2016quantum,zhou2020quantum}.

The QAOA circuit template is, in principle, broadly applicable to combinatorial optimization, with constraints encoded as penalty terms in the cost Hamiltonian. Nevertheless, its standard formulation combines all constraints and objectives from the outset, without explicitly exploiting their hierarchy. The standard mixer explores feasible and infeasible assignments through individual bit flips, making it difficult to concentrate probability on optimal solutions under complex constraints \cite{wang2020xy}. Studies of the corresponding adiabatic dynamics identify weak tunnelling between competing low-energy configurations and small energy gaps as obstacles to reaching the optimum \cite{amin2009first,altshuler2010anderson}. Slow passage through these gaps translates into deep circuits for QAOA schedules obtained by discretizing the adiabatic evolution \cite{farhi2014quantum,jansen2007bounds}. An alternative is to design mixers that preserve specified constraints, keeping the search within the corresponding feasible space \cite{hadfield2019quantum}. For example, XY mixers preserve the number of selected variables and therefore support searches subject to fixed-cardinality constraints \cite{wang2020xy}. More flexible approaches incorporate problem information through initial-state preparation, recursive reduction or the construction of the mixer itself \cite{tate2023warm,bravyi2020obstacles,bartschi2020grover}.

Two developments in quantum algorithm theory motivate combining staged optimization with finite-angle QAOA circuits.  Tunable variable-time amplitude amplification (VTAA) characterizes generic nested amplitude amplification through tunable stage thresholds and deterministic amplification schedules \cite{ambainis2012variable,low2024quantum}.  Its staged reflections distribute amplification effort across stopping times.  Discrete adiabatic theorems bound the error in following an evolving eigenspace through a sequence of slowly varying unitary operations \cite{costa2022optimal}.  Recent results prove self-cancellation of digital adiabatic errors and show that large-step product walks can possess eigenphase gaps distinct from continuous Hamiltonian energy gaps \cite{lu2026digital,an2025large}.  For unstructured search, the same framework supplies QAOA angles with Grover-optimal asymptotic query scaling \cite{an2025large}.  Together, these ideas support a staged QAOA in which an optimized state defines the projector mixer for the next stage; the later circuit then uses finite-angle layers without being restricted to a small-step discretization of continuous evolution. Guided by these principles, we introduce the learned-projector quantum alternating operator ansatz (LP-QAOA) for hierarchical optimization.

\subsection{Overview of LP-QAOA}
\paragraph{Design of LP-QAOA.}
LP-QAOA is a multistage variational algorithm that uses the quantum state learned at one stage to define the mixer for the next.  Each optimized preparation circuit is frozen and reused as the algorithm progresses from structural or coarse objectives to the final objective.

If stage $j$ prepares $\ket{\Phi_{j}}$, the projector onto this learned state is $\Pi_{\Phi_{j}}=\ket{\Phi_{j}}\!\bra{\Phi_{j}}$, and a later learned-projector (LP) mixer uses
\begin{equation}
    R_{\Phi_{j}}=\Id-\Pi_{\Phi_{j}}
    =\Id-\ket{\Phi_{j}}\!\bra{\Phi_{j}}
    \label{eq:history_mixer_intro}
\end{equation}
as its generator. Section~\ref{subsec:main_construction} gives the formal definition and circuit realization, while \suppref{app:related} compares LP-QAOA with constraint-specific mixers, warm starts, Grover-mixer QAOA and other staged QAOA variants.

Our primary focus is on optimization problems with a strict hierarchy of constraints.  Successive constraint levels define nested feasible sets, and the final objective ranks the remaining candidates.  This structure allows LP-QAOA to resolve lower-level requirements first and to reuse the resulting preparation during later refinement.  The same construction can also be applied to a soft hierarchy, in which an earlier objective provides an approximate guide to promising final solutions.  We examine this broader setting as a supporting application.

\paragraph{Stability and energy-gap analysis.}
We establish two analytical results for LP-QAOA. The first bounds how approximation errors propagate through successive frozen stages. It relates the final-state error to the accuracy of intermediate preparations and their subsequent use in learned-projector mixers. This provides a quantitative basis for setting preparation accuracies across the hierarchy (Theorem~\ref{thm:main_stability}).

The second shows how learned-projector mixing avoids the suppression caused by long sequences of local moves between distant feasible configurations. We examine this effect through the adiabatic evolutions associated with local and learned-projector mixers, using their energy gaps to identify obstacles to reaching the optimum. In a hierarchical Exact Cover 3 (EC3) model, local mixing connects distant configurations through many intermediate moves, whereas learned-projector mixing couples them directly. This direct coupling avoids high-order tunnelling suppression and yields a larger minimum energy gap (Theorem~\ref{thm:main_ec3_gap}). The learned state therefore changes both which configurations enter refinement and how the quantum search moves between them.

Together, the staged construction and the analytical results motivate three consequences that organize the evaluation below.  Early stages use cheaper structural Hamiltonians, reserving the expensive final objective for later refinement.  The learned projector couples configurations represented in an earlier state collectively, rather than only through sequences of local moves.  Finite-depth optimization is divided across the hierarchy: each stage optimizes its current circuit before supplying a fixed preparation to the next.  The numerical studies accordingly examine resource use, target concentration and finite-depth trainability.

\paragraph{Numerical results.}
To evaluate LP-QAOA in a strict multilevel setting, we construct a generalized spin-tiling benchmark.  The starting point is the tiling-puzzle formulation of spin-glass instances, in which locally specified subproblems combine into a global constraint-satisfaction problem \cite{hamze2018near}.  We extend this idea to several nested levels of exact constraints followed by a final objective, and call the resulting family block-constrained spin tiling (BCST).  In this benchmark, each LP-QAOA stage corresponds to one level of the strict hierarchy.

Across the BCST instances studied here, LP-QAOA achieves a markedly higher probability of sampling the unique optimum and a lower total cost than the tested two-stage LP-QAOA and block-XY QAOA alternatives.  Ablation tests further show that both the optimized intermediate preparation and its use as the projector reference contribute to these gains.  BCST also provides a common setting for examining trainability through the response of the objective to changes in circuit angles.  At the tested depths, this response is substantially stronger for LP-QAOA than for the local block-XY alternatives, as measured by normalized gradients.

For a complementary study of soft hierarchies, we use the stochastic block model (SBM), originally introduced as a probabilistic model of social networks in which vertices are assigned to groups and connection probabilities depend on those assignments \cite{holland1983stochastic}.  The results show that LP-QAOA can use approximate structural information to improve the probability of finding the optimum in problems with soft hierarchies.

\section{Results}
\label{sec:results}

\subsection{Formal definition and variants of LP-QAOA}
\label{subsec:main_construction}

We define LP-QAOA through its stage-wise variational problems and the construction of the mixers linking successive stages.

Formally, let $\ket{\Phi_0}$ be a fixed, easy-to-prepare initial state for an $m$-stage protocol with $m\ge2$.  The first variational stage ($j=1$) uses a prescribed mixer $M_1$, such as a transverse-field or XY mixer, together with a diagonal phase Hamiltonian $C_1$, a depth $p_1$, and a classical loss functional $\mathcal L_1$.  For each later stage $j\ge2$, choose $C_j$, $p_j$ and $\mathcal L_j$.  Its mixer is constructed from the preceding optimized state as
\begin{equation}
    M_j=R_{\Phi_{j-1}}
    =\Id-\ket{\Phi_{j-1}}\!\bra{\Phi_{j-1}},
    \qquad j\ge2.
    \label{eq:lp_mixer}
\end{equation}
The stage-$j$ circuit is
\begin{equation}
    U_j(\theta_j)
    =
    \prod_{\ell=p_j}^{1}
    \left[
    e^{-i\beta_{j,\ell}M_j}
    e^{-i\gamma_{j,\ell}C_j}
    \right],
    \qquad
    \theta_j=(\beta_{j,\ell},\gamma_{j,\ell})_{\ell=1}^{p_j}.
    \label{eq:stage_ansatz}
\end{equation}
The optimizer solves the stage-local variational problem
\begin{equation}
    \theta_j^\star\in
    \argmin_{\theta_j}
    \mathcal L_j\!\left(
    U_j(\theta_j)\ket{\Phi_{j-1}}
    \right),
    \qquad
    \ket{\Phi_j}=
    U_j(\theta_j^\star)\ket{\Phi_{j-1}}.
    \label{eq:stage_optimization}
\end{equation}
The optimized parameters and preparation circuit for $\ket{\Phi_j}$ are then held fixed when stage $j+1$ is constructed.

\begingroup
\setlength{\abovedisplayskip}{5pt}
\setlength{\belowdisplayskip}{5pt}
\setlength{\abovedisplayshortskip}{0pt}
\setlength{\belowdisplayshortskip}{5pt}
The LP mixer has a Grover-style circuit form, but its reference is learned by previous frozen stages rather than supplied by a prescribed exact reference-state preparation \cite{bartschi2020grover}.  For circuit synthesis, let $P_0\ket{0^N}=\ket{\Phi_0}$, where $P_0$ is a fixed initial-state preparation circuit, and let
\[
    A_{j-1}
    =
    U_{j-1}(\theta_{j-1}^\star)
    \cdots
    U_1(\theta_1^\star)P_0
\]
be the full frozen preparation circuit for $\ket{\Phi_{j-1}}$.  Then an LP-mixer layer can be realized as
\begin{equation}
    e^{-i\beta R_{\Phi_{j-1}}}
    =
    A_{j-1}\,
    e^{-i\beta(\Id-\ket{0^N}\!\bra{0^N})}\,
    A_{j-1}^{\dagger}.
    \label{eq:history_circuit}
\end{equation}
\endgroup

Operationally, the circuit uncomputes the fixed earlier-stage preparation, applies a selective phase to the all-zero state and then re-prepares the reference state.  Figure~\ref{fig:lp_qaoa_nested_circuit} summarizes how the fixed preparation circuit is replayed and uncomputed to implement the learned-projector mixer.

\begingroup
\setlength{\intextsep}{4pt}
\begin{figure}[!ht]
\centering
\captionsetup{font=small,skip=5pt}
\includegraphics[width=0.90\textwidth]{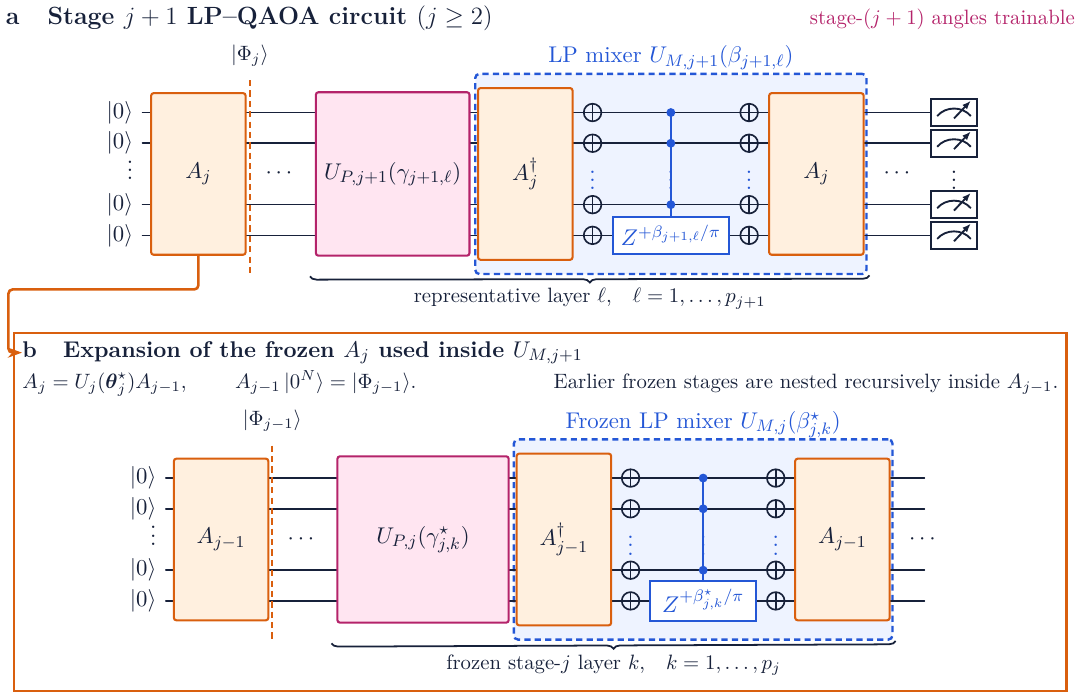}
\caption{\textbf{Learned-projector mixing and recursive nesting in LP-QAOA.}
\textbf{a}, Stage $j+1$ ($j\geq2$) begins with the frozen preparation $A_j$ of $\ket{\Phi_j}$. Up to a global phase, its LP mixer uncomputes with $A_j^{\dagger}$, applies a selective phase to $\ket{0^N}$ and re-prepares with $A_j$. \textbf{b}, Expanding $A_j$ reveals the frozen stage-$j$ circuit, whose LP mixer is constructed from $A_{j-1}$. Starred parameters are fixed.}
\label{fig:lp_qaoa_nested_circuit}
\end{figure}
\FloatBarrier
\endgroup

The construction applies to optimization problems with either strict or soft hierarchical structure.  In the strict case, the stage ground-state sets obey
\begin{equation}
    \calG_m\subseteq\calG_{m-1}\subseteq\cdots\subseteq\calG_1,
    \qquad
    \calG_j=\argmin_z C_j(z).
    \label{eq:strict_refinement}
\end{equation}
If an earlier stage concentrates nearly all of its probability mass in the feasible subspace, later phase separators may omit the corresponding constraint terms.  This reduces the circuit cost of implementing each subsequent phase separator.

In the soft case, exact nesting is replaced by a measurable enrichment criterion: relative to the initial state, the coarse-stage state should assign greater probability to a prespecified low-energy set of the final objective.  Supplementary Algorithm~A0 summarizes the general protocol, and Algorithms~A1--A2 describe the implementations used in our numerical experiments.

LP-QAOA recursively nests the optimized preparation circuits of earlier stages within later circuits. This repeated use of approximate preparations calls for a stability bound on error accumulation, which the following theorem provides.

\begin{theorem}[Frozen-stage stability]
\label{thm:main_stability}
When applied to the ideal preceding state with the ideal mixer, the parameters returned at stage $j$ prepare a chosen ideal target within error $\varepsilon_j$. Let $e_j$ denote the distance between the actual and ideal states after stage $j$. Set $B_j=\sum_{\ell=1}^{p_j}|\beta_{j,\ell}|$ for an LP-mixer stage and $B_j=0$ for a fixed-mixer stage.  Then
\begin{equation}
    e_m\le
    \sum_{j=1}^{m}\varepsilon_j
    \prod_{k=j+1}^{m}(1+B_k).
    \label{eq:main_stability_bound}
\end{equation}
The amplification factors depend on the optimized mixer angles.
\end{theorem}

The factor $1+B_j$ reflects the two ways in which an imperfect frozen preparation affects stage $j$, through its input state and its mixer. The input-state contribution carries forward the preceding error, while each LP-mixer layer contributes at most its absolute angle times that error. The sum $B_j$ therefore controls the mixer contribution, and the products in Eq.~\eqref{eq:main_stability_bound} describe how each stage's error propagates through later stages. For a detailed statement and proof, see \suppref{app:stability}.

\subsection{Mechanisms of LP-QAOA}
\label{subsec:main_mechanism}

We examine how staged optimization and learned-projector mixing contribute to the performance of LP-QAOA.  Our analysis addresses the allocation of quantum resources across objective levels, the coupling of distant feasible configurations and the trainability of the variational circuits.

\paragraph{Resource reallocation.}
Resource reallocation is a central motivation for the staged design of LP-QAOA.  The learned-projector mixer incorporates previously acquired structural information into every refinement layer, allowing the final-stage search to build on the earlier optimization.  Implementing this mixer requires the frozen preparation and its inverse, together with a register-wide selective phase operation.  LP-QAOA thus trades additional mixer resources for a search guided by the learned structure, enabling stronger concentration on the optimum with fewer expensive final-objective applications.

\paragraph{Nonlocal learned-projector mixing.}
Learned-projector mixing provides direct, nonlocal coupling between configurations represented in the learned reference state.  The projector in Eq.~\eqref{eq:lp_mixer} couples any two configurations with nonzero reference amplitudes, with coupling magnitudes given by the products of their amplitude magnitudes.  By contrast, transverse-field mixer Hamiltonians couple configurations differing by one bit, while constraint-preserving XY mixers couple configurations related by an exchange of two qubit occupations \cite{wang2020xy}.  These local Hamiltonians connect distant configurations through sequences of intermediate configurations.  The learned projector supplies a direct coupling between the endpoints, irrespective of their Hamming distance.

We examine the effect of nonlocal mixing through the spectrum of $H(s)=(1-s)M+sC$, which interpolates from the mixer Hamiltonian $M$ to the cost Hamiltonian $C$.  Previous work connects this spectral description both to the depth of QAOA schedules approximating adiabatic evolution \cite{farhi2014quantum,binkowski2024elementary,jansen2007bounds} and to the nonadiabatic dynamics of optimized QAOA \cite{zhou2020quantum}.  We use this connection to compare how local and learned-projector mixers couple competing low-energy configurations during refinement.

We make this comparison in quantum optimization of Exact Cover 3 (EC3), an NP-complete problem in which each clause requires exactly one of three binary variables to equal one.  A solution is an assignment of the variables that satisfies all clauses simultaneously.  EC3 has served as a model for studying localization-induced small gaps in AQO, providing a concrete setting for comparing local and collective mixing \cite{altshuler2010anderson}.  In our hierarchical construction, a subset of disjoint clauses defines the initial feasible set, and refinement optimizes the remaining clauses within this set.  The following theorem gives the learned-projector minimum gap for this refinement.

\begin{theorem}[EC3 feasible-subspace gap]
\label{thm:main_ec3_gap}
Consider an EC3 instance on $N$ bits with polynomially many clauses and a unique satisfying assignment.  Let $F$ be the set of configurations satisfying $m$ disjoint clauses, with $L=|F|=2^N(3/8)^m$.  Define $\ket{F}=L^{-1/2}\sum_{x\in F}\ket{x}$, and let $C_F$ denote the EC3 cost Hamiltonian restricted to the feasible subspace.  On this subspace, the interpolation $H_F(s)=(1-s)(\Id-\ket{F}\!\bra{F})+sC_F$, with $0\le s\le1$, has minimum gap
\begin{equation}
    g_{\min}=q_N L^{-1/2}(1+o(1)),
    \label{eq:main_ec3_gap}
\end{equation}
where $q_N$ lies between an inverse polynomial and a constant.
\end{theorem}

\par\smallskip\noindent
By contrast, Altshuler et al.\ study transverse-field AQO on random EC3 near the point where adding constraints makes satisfying all clauses unlikely.  Their localization analysis predicts minimum-gap scaling $\Delta_{\min}=2^{-\Theta(N\log N)}$ from high-order tunnelling between configurations separated by $\Theta(N)$ bit flips \cite{altshuler2010anderson}.  A block-XY mixer restricts the evolution to the feasible set $F$ but retains local motion on the feasible-configuration graph.  Extending this localization analysis to block-XY mixing predicts a similarly superexponential gap suppression, $g_{\min}^{\mathrm{XY}}\lesssim\exp[-cN\log N+O(N)]$, with $c>0$.  Learned-projector mixing replaces high-order tunnelling through intermediate configurations with direct coupling across the feasible set.  The detailed theorem, proof and local-mixer analysis are given in \suppref{app:gap}.

\paragraph{Trainability improvement.}
Stagewise optimization confines each classical search to the $2p_j$ current-stage angles while keeping earlier parameters fixed.  This keeps each classical search lower-dimensional than a joint optimization of the full nested circuit.

The learned-projector mixer also changes how the loss varies across circuit parameters.  The relevant quantity is the variance of the expected cost over random angles: exponential concentration makes differences between parameter choices increasingly difficult to resolve during optimization.  MaxCut provides a familiar setting for examining this effect.  Two recent results establish a sharp contrast between transverse-field and projector mixing.  In the random-circuit limit reached at sufficiently large depth, standard QAOA on typical random graphs has exponentially small loss variance \cite{mao2025qaoa}.  Grover-mixer QAOA instead retains an inverse-polynomial lower bound for the same problem \cite{tsvelikhovskiy2025provable}.  This contrast identifies projector mixing as a mechanism for avoiding exponential loss concentration.  LP-QAOA combines this mixer structure with a learned reference state and separate optimization of successive stages.  The underlying variance theorems and their Lie-algebraic interpretation are summarized in \suppref{app:trainability_theory}.

\subsection{LP-QAOA on multilevel BCST instances}
\label{subsec:main_bcst}

To test LP-QAOA on a combinatorial optimization problem with a strict hierarchy of constraints, we constructed a block-constrained spin tiling (BCST) problem.  This extends spin-tiling benchmarks, which assemble local spin choices into a global constraint problem \cite{hamze2018near}.  Each site is represented by a block of binary variables, with a local constraint fixing how many labels are selected from that block.  Two further constraints prohibit adjacent sites from selecting the same label and require each label to appear a prescribed total number of times.  A final weighted objective ranks the assignments satisfying all constraints.  Interpreting the labels as radio channels gives a frequency-assignment problem whose final objective combines assignment costs with penalties for third-order intermodulation (IM3) interference.  The complete Hamiltonians, circuit definitions and numerical protocol are given in \suppref{app:bcst_multilevel}. Standard QAOA produced negligible target probability even on a reduced 18-variable instance (\suppref{app:bcst_standard_qaoa}).

Stage 1 of LP-QAOA optimizes the adjacency constraint using a block-XY mixer that preserves the block-cardinality constraint.  Its optimized circuit is frozen and reused in the learned-projector mixer of stage 2, which minimizes the expected global-occupation penalty.  The resulting preparation is then frozen and used to construct the LP mixer for optimizing the final objective.  We also test two-stage LP-QAOA variants that combine global-occupation and final-objective optimization in one learned-projector stage following the adjacency stage, using either separate phase angles or one shared phase angle.  The block-XY QAOA baselines preserve the block-cardinality sector but optimize the remaining constraints and objective directly, using either a combined phase separator or separate phase separators.

\begin{figure}[!ht]
\centering
\includegraphics[width=0.86\linewidth]{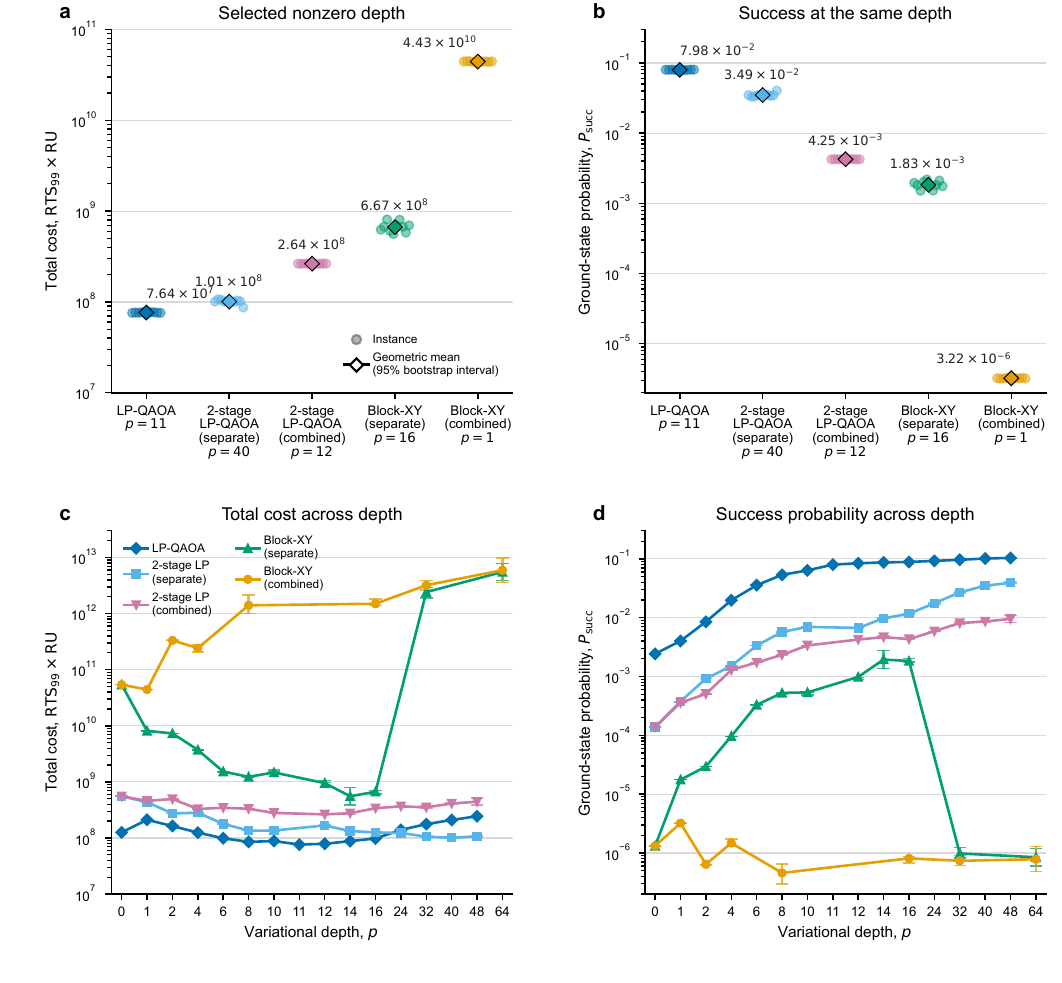}
\caption{Depth-dependent algorithm performance. \textbf{a,b}, Total cost and unique-optimum probability at selected depths. Circles show ten realizations; diamonds show geometric means with 95\% bootstrap intervals. \textbf{c,d}, Depth dependence across the same ten realizations; lines show geometric means and bars show interquartile ranges.}
\label{fig:bcst_main_comparison}
\end{figure}

\FloatBarrier

We tested ten realizations with a common 30-variable constraint structure and independently sampled random coefficients for the final objective.  Each realization had a unique optimum among assignments satisfying all constraints.  The frozen structural preparations were therefore shared across the ten realizations.  Every nonzero-depth circuit for final-objective optimization was optimized independently with a matched evaluation budget and four parameter initializations.  At each method's selected nonzero depth, all three LP-QAOA variants achieved lower total cost and higher unique-optimum probability than either block-XY baseline on all ten realizations.  Compared with separate-phase block-XY, the lower-cost baseline, LP-QAOA and its two-stage variants with separate and combined phases reduced geometric-mean total cost by factors of 8.73, 6.61 and 2.53, respectively.  Three-stage LP-QAOA also outperformed both two-stage variants on both metrics for every realization (Fig.~\ref{fig:bcst_main_comparison}a,b).  LP-QAOA's unique-optimum probability continued to rise at larger $p$, but total cost reached a broad minimum at $p=8$--16.  At greater depths, the increased logical cost per circuit outweighed the improvement in success probability (Fig.~\ref{fig:bcst_main_comparison}c,d).

Ablation tests on the same ten realizations examined the contributions of the intermediate preparation and its use in the learned-projector mixer.  We omitted the second structural stage, replaced learned-projector refinement with a warm-start local mixer, or constructed the projector from the initial fixed-cardinality state.  LP-QAOA achieved lower total cost and higher unique-optimum probability than each of these alternatives on all ten realizations (Fig.~\ref{fig:bcst_component_tests}).  These comparisons link the performance advantage to the optimized intermediate preparation and its later use as the projector reference.

\begin{figure}[!t]
\centering
\includegraphics[width=0.98\linewidth]{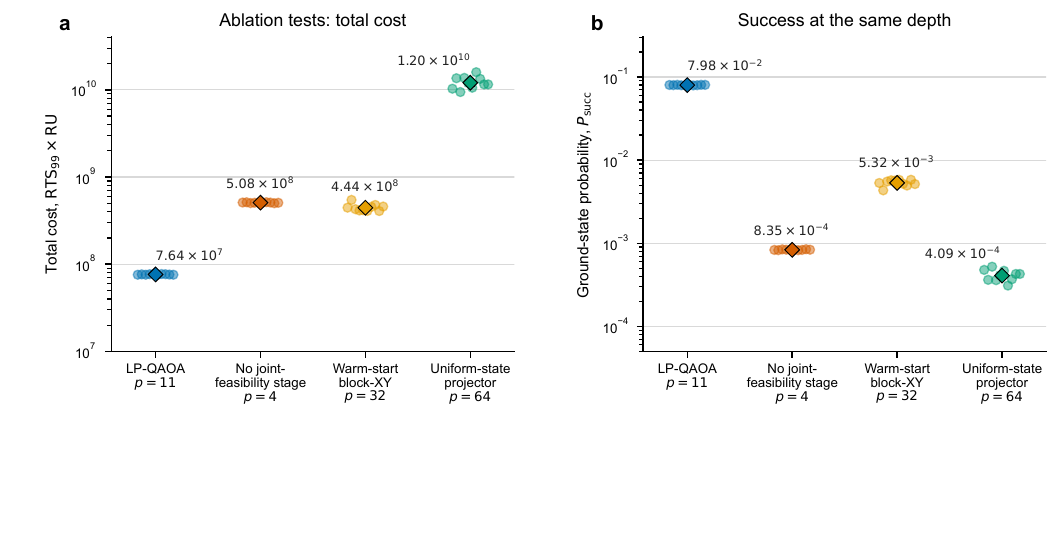}
\caption{Ablation tests of LP-QAOA. \textbf{a}, Total cost at each construction's selected nonzero depth. \textbf{b}, Unique-optimum probability at the same depths. Circles show individual final-objective realizations; diamonds show geometric means, with 95\% bootstrap intervals over ten realizations.}
\label{fig:bcst_component_tests}
\end{figure}

\FloatBarrier

\paragraph{Finite-depth trainability.}
We evaluated gradients at randomly sampled circuit angles on a separate set of six $N=30$ BCST instances to assess finite-depth trainability.  At depths 26, 29 and 32, LP-QAOA had normalized gradient magnitudes approximately 70--90 times larger than three block-XY QAOA baselines.  This held for both the training loss and final objective, after normalization for parameter count and Hamiltonian energy scale.  Its loss-decreasing direction was more closely aligned with increasing unique-optimum probability in seven of the nine comparisons across three baselines and three depths (Fig.~\ref{fig:gradient_trainability}).  These results show stronger sensitivity to parameter changes and better alignment between reducing the loss and increasing success probability.  For definitions and full results, see \suppref{app:gradient_diagnostics}.

\begin{figure}[!htbp]
\centering
\includegraphics[width=0.98\linewidth]{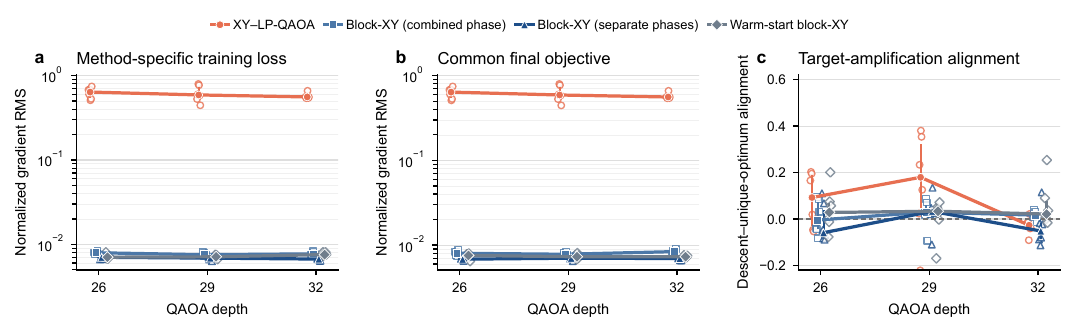}
\caption{BCST gradient diagnostics on six $N=30$ instances.  \textbf{a,b}, Normalized gradient magnitudes for the training loss and final objective.  \textbf{c}, Alignment between loss descent and increasing unique-optimum probability.  Open symbols show instance medians; filled symbols and bars show the median and interquartile range across instances.}
\label{fig:gradient_trainability}
\end{figure}

\subsection{Soft coarse-to-fine optimization on SBM instances}
\label{subsec:main_sbm}

As a complementary test, we asked whether LP-QAOA could use a soft hierarchy that retains coarse statistical structure rather than imposing exact constraints.  We studied nine $N=24$ two-community stochastic block model (SBM) graphs with $(p_{\mathrm{in}},p_{\mathrm{out}})=(0.80,0.36)$ and constructed a coarse Hamiltonian that favours agreement within four-vertex groups and uses averaged couplings between groups \cite{holland1983stochastic}.  On each of nine graphs, LP-QAOA assigned more probability to the exact ground states than the tested LP stage-1 objective QAOA, warm-start QAOA and standard QAOA circuits.  The median probability was $0.777$, compared with $0.293$ for LP stage-1 objective QAOA and $0.434$ for the best tested standard QAOA depth (Fig.~\ref{fig:sbm_adam_comparison}).  The higher ground-state probability outweighed the additional circuit cost, giving LP-QAOA a lower median total cost than all tested alternatives.  The complete construction and comparisons are given in \suppref{app:soft_sbm_csp}.

\begin{figure}[!htbp]
\centering
\includegraphics[width=0.98\linewidth]{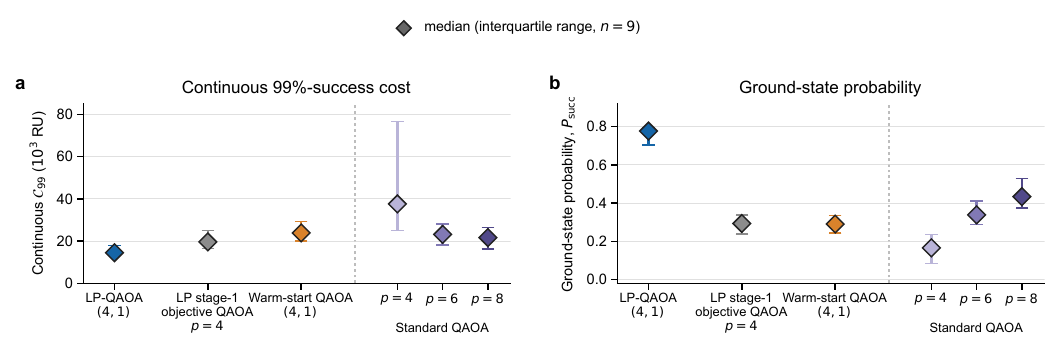}
\caption{\textbf{Learned-projector refinement on SBM instances.} \textbf{a}, Median continuous cost for 99\% success across nine graphs. \textbf{b}, Median ground-state probability across the same graphs. Bars show interquartile ranges.}
\label{fig:sbm_adam_comparison}
\end{figure}

\FloatBarrier

\section{Discussion}
\label{sec:discussion}

LP-QAOA uses intermediate optimization results to change the dynamics of subsequent stages. By learning the reference state of a projector mixer, it makes structural information part of the operators that guide the later search. This connects LP-QAOA to the broader development of constraint-preserving mixers, warm starts and Grover-style projector mixing \cite{hadfield2019quantum,tate2023warm,bartschi2020grover}. Its distinguishing feature is the recursive use of variationally optimized preparations to construct successive mixers. A detailed comparison with related approaches is given in \suppref{app:related}.

The theoretical and numerical results connect this construction to resource reallocation, collective coupling and improved trainability. In BCST, LP-QAOA achieved higher success probability with fewer final-objective applications, and the reduction in repeated circuit executions outweighed the additional mixer cost. These results demonstrate the benefit of investing resources in learned mixing to reduce the effort required for the final search. The collective coupling examined in the EC3 analysis offers a physical interpretation of the stronger target concentration in BCST. The state learned before final-objective optimization has support on many feasible configurations, allowing the LP mixer to connect configurations that local mixers reach only through successive moves. The gradient diagnostics also connect the choice of mixer to classical optimization: larger gradients, more closely aligned with increasing success probability, favour parameter optimization and may help LP-QAOA avoid barren plateaus. This finite-depth evidence complements the asymptotic variance contrast between transverse-field and uniform-projector mixing. The supporting SBM experiment extends the use of learned structural information to problems with soft hierarchical structure.

Realizing these benefits on hardware involves a trade-off between the complexity of each mixing operation and the resources needed to obtain the optimum. Each LP layer uncomputes and replays an earlier preparation and applies a selective phase across the full register. The resulting circuit depth is demanding for noisy intermediate-scale quantum (NISQ) devices. Experimental advances in logical quantum processing and below-threshold error correction are opening a path towards early fault-tolerant computation \cite{bluvstein2026fault}. This development motivates algorithms that use longer reliable circuits effectively while keeping logical-resource requirements manageable. LP-QAOA offers a concrete strategy for this regime by exchanging more elaborate mixing operations for fewer expensive objective evaluations and repeated preparations. A complementary route towards shallower implementations is to construct projector mixers from states learned on overlapping subsets of variables. Varying the size of these subsets would allow the balance between nonlocal coupling, preparation depth and target concentration to be explored for different hardware capabilities.

In summary, LP-QAOA turns a hierarchy of objectives into a sequence of variationally learned mixing operations. The stability and spectral analyses provide a theoretical basis for recursive state reuse and collective search, while the numerical experiments demonstrate higher success probabilities, lower logical costs and improved finite-depth trainability. Together, these results show how information acquired during optimization can be used to improve both the subsequent quantum search and its classical parameter optimization.

\section{Methods}
\paragraph{Stability analysis.}
We use Duhamel's formula to bound changes in LP-mixer evolutions caused by errors in their reference states. Combining these bounds across circuit layers and propagating them through successive stages yields the stability theorem. The detailed statement and proof are given in \suppref{app:stability}.

\paragraph{Spectral-gap analysis.}
For the LP mixer, we reduce the EC3 interpolation within the feasible subspace to a diagonal operator with a rank-one perturbation. We obtain its two lowest eigenvalues from the secular equation and expand near their avoided crossing to determine the minimum gap. For block-XY mixing, we adapt the localization argument of Altshuler et al.\ to the graph of allowed configuration changes \cite{altshuler2010anderson}. Graph distance sets the minimum perturbative order of tunnelling, and the localization path-sum estimate bounds the coupling amplitude. The adiabatic theorem and product-formula error bounds connect the spectrum to the depth of QAOA circuits approximating the evolution \cite{jansen2007bounds,farhi2014quantum}. Full derivations are given in \suppref{app:gap}.

\paragraph{Numerical implementation and evaluation.}
All numerical experiments used state-vector simulation. BCST circuits preserving block cardinality were simulated in the complete corresponding subspace. The standard-QAOA diagnostic and SBM circuits were simulated in the full computational basis. The main BCST and SBM comparisons used Adam for angle optimization; the supplementary BCST robustness analysis used simultaneous perturbation stochastic approximation (SPSA). Performance was evaluated by target probability and the estimated logical cost of obtaining a target with 99\% probability. BCST used \(\mathrm{RTS}_{99}\times\mathrm{RU}\), whereas SBM used its continuous relaxation in Eq.~\eqref{eq:sbm_continuous_cost}. Resource accounting includes the complete circuit and terminal verification, while excluding classical training and hardware-dependent overhead. Detailed methods and experimental settings are provided in the Supplementary Information and Source Data.

\paragraph{Use of artificial intelligence.}
The authors conceived LP-QAOA and established its theoretical results. ChatGPT and Codex assisted with implementing and checking simulation code, preparing figures and editing the manuscript. The authors reviewed the code, results, figures and text and take responsibility for the work.

\section{Data and code availability}

The data and code supporting this study are available in the public \href{https://github.com/kkzky/Nested-QAOA}{project repository}. The repository also provides the software environment required to reproduce the reported analyses.

\section*{Acknowledgements}
J.-P.L. acknowledges support from the Quantum Science and Technology National Science and Technology Major Project under Grant No.~2024ZD0300500, the Excellent Young Scientists Fund Program, start-up funding from Tsinghua University, and the Beijing Institute of Mathematical Sciences and Applications.

We also thank Ningfeng Wang and Zhengwei Liu for insightful and valuable discussions.

\printbibliography[title={References}]
\fi

\ifsupplementcontent
\ifdefined\supplementaryonly
\maketitle
\makeatletter
\renewcommand*\l@section{\@dottedtocline{1}{0em}{3.5em}}
\renewcommand*\l@subsection{\@dottedtocline{2}{1.5em}{4.5em}}
\makeatother
\tableofcontents
\clearpage
\fi
\appendix
\ifdefined\supplementaryonly
\renewcommand{\thesection}{S\arabic{section}}
\counterwithout{equation}{section}
\setcounter{equation}{0}
\renewcommand{\thedefinition}{S\arabic{section}.\arabic{definition}}
\renewcommand{\theassumption}{S\arabic{section}.\arabic{assumption}}
\renewcommand{\thetheorem}{S\arabic{section}.\arabic{theorem}}
\renewcommand{\thelemma}{S\arabic{section}.\arabic{lemma}}
\renewcommand{\theproposition}{\thetheorem}
\renewcommand{\theclaim}{\thetheorem}
\renewcommand{\theremark}{S\arabic{section}.\arabic{remark}}
\renewcommand{\thetable}{\arabic{table}}
\renewcommand{\thefigure}{\arabic{figure}}
\renewcommand{\theequation}{S\arabic{equation}}
\renewcommand{\tablename}{Supplementary Table}
\renewcommand{\figurename}{Supplementary Fig.}
\fi

\section{Supplementary discussion: relation to existing QAOA variants}
\label{app:related}

\paragraph{Standard QAOA and alternating operators.}
The standard quantum approximate optimization algorithm (QAOA) uses one diagonal objective Hamiltonian and a transverse-field mixer, usually optimizing the expected objective value \cite{farhi2014quantum}.  The quantum alternating operator ansatz generalizes the mixer and phase separator choices \cite{hadfield2019quantum}.  The learned-projector quantum alternating operator ansatz (LP-QAOA) adds an inter-stage dependency to this ansatz: later operators may depend on projectors learned from earlier optimized states.

\paragraph{Constraint-specific mixers.}
Constraint-preserving mixers restrict QAOA dynamics to a subspace fixed by known constraints \cite{hadfield2017quantum,wang2020xy,fuchs2022constraint}. XY mixers do so by preserving Hamming weight, making them natural for one-hot and fixed-cardinality encodings. Logical-X QAOA (LX-QAOA) generalizes this construction: stabilizer code spaces and logical \(X\) rotations provide mixers for a broader class of prescribed subspaces \cite{fuchs2024lx}. These mixers are tailored to constraints known before optimization, and their applicability depends on an efficient encoding and connectivity within the prescribed feasible subspace. In block-constrained spin tiling (BCST), block-XY mixing preserves the prescribed number of selected labels in each block (two in the present benchmark), whereas the stage-1 conflict Hamiltonian distinguishes assignments that violate cross-block conflicts. XY-QAOA and LX-QAOA therefore derive their mixers from structure known before optimization; LP-QAOA instead constructs a later mixer from the projector onto a state learned in an earlier variational stage.

\paragraph{Grover-mixer QAOA.}
Grover-mixer QAOA assumes a known state-preparation unitary $U_S$ for a prescribed reference state $\ket{F}$ and uses the projector mixer $U_S e^{-i\beta\ket{0^N}\!\bra{0^N}}U_S^\dagger$ \cite{bartschi2020grover}.  The LP mixer has the same rank-one-projector circuit form, but its reference state is produced by earlier variational stages and every layer therefore replays and uncomputes the frozen optimized preparation.  Grover-mixer QAOA does not replay an earlier optimized circuit.  When an efficient exact preparation of the reference state is known, Grover-mixer QAOA requires no preceding variational stage.

\paragraph{Warm-started QAOA.}
Warm-started QAOA prepares a nonuniform or problem-informed initial state before applying alternating phase and mixer operators \cite{tate2023warm}.  LP-QAOA additionally inserts the projector onto the earlier state into the later mixer.  In both the BCST and stochastic block model (SBM) comparisons, the warm-start baseline reuses the earlier state solely as an input, whereas LP-QAOA also uses it to define the later mixer.

\paragraph{Alternative losses based on conditional value at risk.}
Conditional value at risk (CVaR) optimization changes the classical loss to emphasize low-energy samples \cite{barkoutsos2020improving}.  It is compatible with different variants of QAOA, including LP-QAOA, while leaving the learned inter-stage projector unchanged.

\paragraph{Recursive, iterative and multilevel optimization.}
Recursive QAOA and quantum-informed recursive optimization extract correlations from a quantum state and use them in classical reduction rules that shrink the problem instance \cite{bravyi2020obstacles,finzgar2024qiro}.  The iterative many-body-localization protocol measures after each cycle and updates the reference Hamiltonian for the next cycle \cite{wang2022mbl}.  Multilevel QAOA applies QAOA within a hierarchy of classically coarsened graph problems and transfers information between resolutions \cite{bach2024mlqaoa}.  Unlike these measurement-based or classically mediated updates, LP-QAOA coherently replays the optimized stage-preparation circuit within a later mixer on the same variable register. This gives the refinement circuit coherent access to the learned amplitudes and phases, while adding the costs of preparation replay and uncomputation.

\paragraph{Constraint-oriented QAOA.}
Recent constraint-oriented QAOA studies include quantum tree-generator state preparation, partitioned-constraint QAOA, iterative warm-start XY-mixer QAOA, locally acting Grover-mixer QAOA and initial-state--mixer alignment \cite{christiansen2025quantum,wilkie2026partitioned,bucher2026constrained,choi2026locally,he2023alignment}. They share the same interest in exploiting the constraint structure of various combinatorial problems to improve the efficiency of QAOA\@. LP-QAOA instead links multiple stages through a learned-projector mixer and stage-specific phase separators.

\paragraph{Discrete adiabatic evolution.}
Discrete adiabatic theorems provide a complementary interpretation of alternating finite-step unitaries \cite{costa2022optimal}. Digital error cancellation supports accurate circuit discretization under effective-gap conditions \cite{lu2026digital}, and large-step product walks can follow eigenphase paths distinct from continuous Hamiltonian eigenpaths \cite{an2025large}. These results motivate the spectral diagnostic in \suppref{app:gap}.

\section{Algorithmic variants}
\label{app:variants}

This section presents the general LP-QAOA protocol and the two benchmarked variants.

\begin{center}
\fbox{%
\begin{minipage}{0.94\textwidth}
\textbf{Algorithm A0: Frozen-stage LP-QAOA.}

\textbf{Input:} An easy-to-prepare initial state $\ket{\Phi_0}$; stage depths $p_1,\ldots,p_m$ with $m\ge2$; a prescribed first-stage mixer $M_1$; and, for each stage $j$, a diagonal phase Hamiltonian $C_j$ and a classical loss functional $\mathcal L_j(\ket{\psi})$.  For every $j\ge2$, the mixer is $M_j=R_{\Phi_{j-1}}$.

\begin{enumerate}[leftmargin=2em]
    \item Initialize the register in $\ket{\Phi_0}$.
    \item For $j=1,\ldots,m$:
    \begin{enumerate}[leftmargin=2em]
        \item Set $M_j$: use the prescribed mixer $M_1$ for $j=1$; for $j\ge2$, use the LP mixer $R_{\Phi_{j-1}}$ constructed from the optimized previous state.
        \item Build the alternating stage circuit $U_j(\theta_j)$ in Eq.~\eqref{eq:stage_ansatz}.
        \item Optimize the stage-local loss in Eq.~\eqref{eq:stage_optimization}.  The loss may be an expectation value, a CVaR loss, a feasibility loss or another stage-specific objective.
        \item Freeze the optimized parameters and preparation circuit for stage $j$. Re-prepare $\ket{\Phi_j}$ when constructing or executing the next stage.
    \end{enumerate}
    \item Output $\ket{\Phi_m}$ and sample in the computational basis.
\end{enumerate}
\end{minipage}}
\end{center}

Equation~\eqref{eq:stage_ansatz} uses one phase Hamiltonian per stage.  A more general independently tunable form replaces
\[
    e^{-i\gamma_{j,\ell}C_j}
    \quad\longrightarrow\quad
    \prod_{a=1}^{q_j}e^{-i\gamma_{j,\ell,a}C_{j,a}},
\]
with stage parameters $\theta_j=(\beta_{j,\ell},\gamma_{j,\ell,a})_{\ell,a}$.

\begin{center}
\fbox{%
\begin{minipage}{0.94\textwidth}
\textbf{Algorithm A1: LP-QAOA for multilevel BCST.}
\begin{enumerate}[leftmargin=2em]
    \item Prepare the product of uniform blockwise fixed-cardinality states, with two excitations per block for the present BCST benchmark.
    \item Stage 1: use the block-XY mixer, comprising all pairwise XY couplings within each block, and the conflict Hamiltonian:
    \[
        U_1^{\mathrm{XY}}=
        \prod_{q=p_1}^{1}e^{-i\beta_{1,q}M_{\mathrm{XY}}}e^{-i\gamma_{1,q}C_{\mathrm{conf}}}.
    \]
    Optimize the conflict-feasibility loss and fix the resulting preparation circuit, which defines $\ket{\Phi_1}=U_1^{\mathrm{XY}}\ket{\Phi_0}$.
    \item Stage 2: let $Q$ be the global-occupation penalty defined in \suppref{app:bcst_multilevel}.  Define $R_{\Phi_1}=\Id-\ket{\Phi_1}\!\bra{\Phi_1}$ and minimize the expectation value of $Q$ using
    \[
        U_2=
        \prod_{q=p_2}^{1}e^{-i\beta_{2,q}R_{\Phi_1}}e^{-i\gamma_{2,q}Q}.
    \]
    Fix this circuit and define $\ket{\Phi_2}=U_2\ket{\Phi_1}$.
    \item Stage 3: define $R_{\Phi_2}=\Id-\ket{\Phi_2}\!\bra{\Phi_2}$ and optimize
    \[
        U_3=
        \prod_{q=p_3}^{1}e^{-i\beta_{3,q}R_{\Phi_2}}e^{-i\gamma_{3,q}C_{\mathrm{obj}}}.
    \]
    \item Sample from $\ket{\Phi_3}=U_3\ket{\Phi_2}$ and evaluate feasibility and optimality.
\end{enumerate}
\end{minipage}}
\end{center}

The two-stage variants used in the multilevel BCST study combine global-occupation and objective optimization in one stage: they start from $\ket{\Phi_1}$, retain the mixer $R_{\Phi_1}$, and assign either independent phase angles to $Q$ and $C_{\mathrm{obj}}$ or a shared angle to their sum in each layer.  The additional two-level BCST study in \suppref{app:supplementary_numerics} instead applies objective refinement directly after its conflict stage.

\begin{center}
\fbox{%
\begin{minipage}{0.94\textwidth}
\textbf{Algorithm A2: LP-QAOA for a soft coarse-to-fine problem.}
\begin{enumerate}[leftmargin=2em]
    \item Let $C_{\mathrm{coarse}}$ be a coarse proxy and $C_{\mathrm{fine}}$ be the final objective.
    \item Stage 1: optimize a coarse QAOA or alternating-operator circuit using a coarse loss, producing $\ket{\Phi_1}$.
    \item Stage 2: use $R_{\Phi_1}$ as the mixer and the final Hamiltonian as the phase separator.
    \item Optimize the final-stage parameters using the expectation value of the final Hamiltonian.
    \item Output and sample the final state.
\end{enumerate}
\end{minipage}}
\end{center}

\section{Stability theorem for frozen-stage projector reuse}
\label{app:stability}

This section gives the formal statement and proof of Theorem~\ref{thm:main_stability}.

For normalized states, define the phase-insensitive distance
\begin{equation}
    d(\psi,\phi)=
    \min_{\vartheta\in\R}
    \norm{\ket{\psi}-e^{i\vartheta}\ket{\phi}}.
    \label{eq:phase_distance}
\end{equation}
For a strict refinement sequence, let $\calH_j^\star$ be the target subspace at stage $j$, usually the span of $\{\ket{z}:z\in\calG_j\}$.  Choose ideal states
\begin{equation}
    \ket{\Phi_j^\star}\in\calH_j^\star,
    \qquad j=0,1,\ldots,m,
\end{equation}
with $\ket{\Phi_0^\star}=\ket{\Phi_0}$.  For every LP-mixer stage, define the ideal and actual complementary mixer projectors
\begin{equation}
    R_j^\star=\Id-\ket{\Phi_{j-1}^\star}\!\bra{\Phi_{j-1}^\star},
    \qquad
    R_j=\Id-\ket{\Phi_{j-1}}\!\bra{\Phi_{j-1}}.
\end{equation}
The stage error is $e_j=d(\Phi_j,\Phi_j^\star)$.

\begin{lemma}
\label{lem:projector_lipschitz}
For normalized $\ket{u}$ and $\ket{v}$,
\begin{equation}
    \norm{\ket{u}\!\bra{u}-\ket{v}\!\bra{v}}
    =\sqrt{1-|\braket{u}{v}|^2}
    \le d(u,v).
    \label{eq:projector_exact_distance}
\end{equation}
Consequently, $\norm{R_j-R_j^\star}\le e_{j-1}$.
\end{lemma}

\begin{proof}
Both projectors vanish outside $\operatorname{span}\{\ket{u},\ket{v}\}$.  Direct diagonalization in this two-dimensional span gives the equality in Eq.~\eqref{eq:projector_exact_distance}.  If $c=|\braket{u}{v}|$, then $d(u,v)^2=2(1-c)$ after phase alignment, while $1-c^2=(1-c)(1+c)\le2(1-c)$.  Taking square roots gives the inequality.
\end{proof}

\begin{lemma}[Learned-projector-mixer error]
\label{lem:stage_lipschitz}
Let
\begin{equation}
    U_R(\beta,\gamma)=\prod_{\ell=p}^{1}e^{-i\beta_\ell R}e^{-i\gamma_\ell C},
\end{equation}
where products are ordered with larger $\ell$ to the left and $R,R'$ are Hermitian.  Then
\begin{equation}
    \norm{U_R(\beta,\gamma)-U_{R'}(\beta,\gamma)}
    \le
    \left(\sum_{\ell=1}^{p}|\beta_\ell|\right)\norm{R-R'}.
\end{equation}
\end{lemma}

\begin{proof}
Duhamel's formula gives
\begin{equation}
    \norm{e^{-i\beta R}-e^{-i\beta R'}}
    \le |\beta|\norm{R-R'}.
\end{equation}
A telescoping expansion over the $p$ mixer factors, together with unitary invariance of the operator norm, gives the stated bound.
\end{proof}

\begin{assumption}[Strict hierarchical objectives]
\label{ass:strict_reachability}
The stage Hamiltonians satisfy Eq.~\eqref{eq:strict_refinement}, and each ideal state $\ket{\Phi_j^\star}$ lies in its corresponding target subspace $\calH_j^\star$.

The strict hierarchy supplies the nested target subspaces used to define the ideal preceding-stage states; the norm bound then depends on the distances to these ideal states and on the mixer-angle budgets.
\end{assumption}

\begin{theorem}[Conditional stability of strict LP-QAOA]
\label{thm:error_propagation}
Under Assumption~\ref{ass:strict_reachability}, let $\theta_j$ denote the parameters returned by the optimizer at stage $j$.  For each stage $j$, suppose these parameters prepare a state within \(\varepsilon_j\) of the ideal target:
\begin{equation}
    d\!\left(
        U_j(R_j^\star;\theta_j)\ket{\Phi_{j-1}^\star},
        \ket{\Phi_j^\star}
    \right)\le\varepsilon_j.
    \label{eq:returned_angle_accuracy}
\end{equation}
For a fixed-mixer stage, the generator argument in Eq.~\eqref{eq:returned_angle_accuracy} is omitted.  For every LP-mixer stage, define
\begin{equation}
    B_j=\sum_{\ell=1}^{p_j}|\beta_{j,\ell}|
\end{equation}
and set $B_j=0$ for every fixed-mixer stage.  The stage-by-stage LP-QAOA output satisfies
\begin{equation}
    e_j\le (1+B_j)e_{j-1}+\varepsilon_j,
    \qquad e_0=0.
    \label{eq:error_recurrence_product}
\end{equation}
Consequently,
\begin{equation}
    e_m\le \sum_{j=1}^m \varepsilon_j
    \prod_{k=j+1}^m(1+B_k).
    \label{eq:final_error_product}
\end{equation}
\end{theorem}

\begin{proof}
Write $U_j(R)=U_j(R;\theta_j)$ for the circuit with the parameters returned at stage $j$.  Choose phase-aligned representatives of $\ket{\Phi_{j-1}}$ and $\ket{\Phi_{j-1}^\star}$ whose distance is $e_{j-1}$.  For an LP-mixer stage, insert the ideal-input and ideal-mixer evolutions and apply the triangle inequality:
\begin{align}
    e_j
    &\le \norm{U_j(R_j)\ket{\Phi_{j-1}}-U_j(R_j)\ket{\Phi_{j-1}^\star}} \\
    &\quad +\norm{U_j(R_j)\ket{\Phi_{j-1}^\star}-U_j(R_j^\star)\ket{\Phi_{j-1}^\star}} \\
    &\quad +\norm{U_j(R_j^\star)\ket{\Phi_{j-1}^\star}-\ket{\Phi_j^\star}}.
\end{align}
The first term equals $e_{j-1}$ by unitarity.  The condition in Eq.~\eqref{eq:returned_angle_accuracy} bounds the third term by $\varepsilon_j$.  Lemmas~\ref{lem:projector_lipschitz} and~\ref{lem:stage_lipschitz} bound the second term by $B_j e_{j-1}$.  A fixed-mixer stage has the same bound with $B_j=0$.  This proves Eq.~\eqref{eq:error_recurrence_product}; induction over $j$ gives Eq.~\eqref{eq:final_error_product}.
\end{proof}

\section{Learned-projector mixer gap theorem for Exact Cover 3 (EC3)}
\label{app:gap}

This section connects the runtime of adiabatic quantum optimization (AQO) to the depth of the corresponding QAOA circuit.  Section~\ref{app:gap_aqo} uses the adiabatic theorem and a product-formula discretization to translate the gap-dependent AQO runtime bound into the depth scale of this circuit construction.  Section~\ref{app:gap_collective} derives the minimum-gap scaling of an ideal learned-projector mixer on the EC3 feasible subspace.  Section~\ref{app:gap_local} adapts the random-EC3 Anderson-localization argument to the block-XY feasible graph and obtains the corresponding superexponential gap estimate within the same random-instance localization picture.

\subsection{From adiabatic runtime to QAOA circuits}
\label{app:gap_aqo}

The minimum gap enters the analysis through adiabatic quantum optimization.  Consider
\begin{equation}
    H(s)=(1-s)M+sC,\qquad s\in[0,1].
\end{equation}
Let $H(s)\ket{E_k(s)}=E_k(s)\ket{E_k(s)}$, with a nondegenerate ground-state branch and gap $g(s)=E_1(s)-E_0(s)>0$.  Set $s=t/T$ and $\Pi_0(s)=\ket{E_0(s)}\!\bra{E_0(s)}$.  The propagator $U_T(s)$ satisfies $i\partial_sU_T(s)=TH(s)U_T(s)$ with $U_T(0)=\Id$.  For this linear interpolation, $H'(s)=C-M$ and $H''(s)=0$.  Writing $h=\norm{C-M}$, the quantitative adiabatic theorem bounds the final ground-state leakage amplitude by \cite{jansen2007bounds}
\begin{equation}
    \norm{(\Id-\Pi_0(1))U_T(1)\Pi_0(0)}
    \le\frac{\mathcal A_H}{T},
    \qquad
    \mathcal A_H=
    \frac{h}{g(0)^2}+\frac{h}{g(1)^2}
    +7h^2\int_0^1\frac{ds}{g(s)^3}.
    \label{eq:adiabatic_runtime_condition}
\end{equation}
Choosing $T\ge\mathcal A_H/\epsilon$ therefore makes this amplitude at most $\epsilon$ and the corresponding transition probability at most $\epsilon^2$.  The gap profile enters through both endpoint terms and the integral.  Its minimum $g_{\min}=\min_s g(s)$ gives the bound $\mathcal A_H\le2h/g_{\min}^2+7h^2/g_{\min}^3$, while retaining $g(s)$ resolves the width of the spectral bottleneck \cite{jansen2007bounds,albash2018adiabatic}.

To connect AQO difficulty with QAOA circuit resources, we use the mixer--cost dynamics shared by the two algorithms.  AQO evolves under the weighted Hamiltonian $H(s)$, whereas QAOA alternates evolutions generated by the same operators $M$ and $C$ \cite{farhi2014quantum,brady2021optimal}.  Divide the runtime into $p$ steps of width $\Delta t=T/p$ and sample the schedule at $s_\ell=(\ell-1)/p$.  The first-order product-formula approximation is the QAOA circuit
\begin{align}
    U_p
    &=
    \prod_{\ell=p}^{1}
    e^{-i\Delta t(1-s_\ell)M}
    e^{-i\Delta t s_\ell C}, \notag\\
    \beta_\ell&=\Delta t(1-s_\ell),
    \qquad
    \gamma_\ell=\Delta t s_\ell.
    \label{eq:adiabatic_qaoa_angles}
\end{align}
Let $\kappa=\norm{[M,C]}$.  Sampling the linear schedule contributes at most $hT/(2p)$ to the operator-norm error.  Splitting each sampled Hamiltonian contributes at most $\Delta t^2s_\ell(1-s_\ell)\kappa/2\le\kappa\Delta t^2/8$ per step.  A telescoping bound over the unitary steps therefore gives
\begin{equation}
    \norm{U_T(1)-U_p}
    \le\frac{hT}{2p}+\frac{\kappa T^2}{8p}.
    \label{eq:adiabatic_discretization_error}
\end{equation}
For a prescribed discretization error $\epsilon_{\mathrm{pf}}$, it is sufficient to choose $p\ge(hT/2+\kappa T^2/8)/\epsilon_{\mathrm{pf}}$.  The final ground-state failure probability is then at most $(\mathcal A_H/T+\epsilon_{\mathrm{pf}})^2$.  Equations~\eqref{eq:adiabatic_runtime_condition} and~\eqref{eq:adiabatic_discretization_error} thus relate the gap profile and Hamiltonian norms to a QAOA depth at a specified accuracy.

The digital circuit has a corresponding spectral description.  Define the step unitary
\begin{equation}
    W_\ell=e^{-i\beta_\ell M}e^{-i\gamma_\ell C}.
\end{equation}
The eigenvalues of $W_\ell$ lie on the unit circle.  The shortest angular separation between the tracked eigenphase branch and any other branch defines the eigenphase gap $\delta_\ell$.  Discrete adiabatic theorems bound the tracking error through these gaps and the variation between successive step unitaries \cite{costa2022optimal}.  Choose $\Delta t$ so that $\Delta t\max_s[E_{\max}(s)-E_0(s)]<\pi$, where $E_{\max}(s)$ is the highest eigenvalue of $H(s)$.  The eigenphases of $e^{-i\Delta t H(s_\ell)}$ then occupy an arc shorter than $\pi$, and the small-step relations are
\begin{equation}
    W_\ell=e^{-i\Delta t H(s_\ell)}+O(\kappa\Delta t^2),
    \qquad
    \delta_\ell=\Delta t\,g(s_\ell)+O(\kappa\Delta t^2).
    \label{eq:step_eigenphase_gap}
\end{equation}
Taking $\kappa\Delta t\ll g_{\min}$ resolves the tracked branch relative to the product-formula correction.  As the schedule grid is refined, $\min_\ell g(s_\ell)$ approaches $g_{\min}$.  The energy-gap bottleneck of AQO therefore appears in the eigenphase spectrum of this small-step circuit.

Sections~\ref{app:gap_collective} and~\ref{app:gap_local} compare these spectral bottlenecks for EC3.  In the localized-crossing description, local mixing connects distant low-energy configurations through high-order paths, whereas learned-projector mixing supplies collective coupling across the feasible subspace.  The construction above connects this spectral comparison to QAOA schedules approximating the corresponding adiabatic evolutions.

\subsection{Feasible-subspace reduction for EC3}
\label{app:gap_collective}

EC3 is an NP-complete constraint satisfaction problem in which each clause acts on three bits and is satisfied only when exactly one bit has value one \cite{schaefer1978complexity}.  Its worst-case classical hardness and established use in localization analyses of AQO make EC3 a useful test case for comparing local and learned-projector mixers \cite{altshuler2010anderson}.  For an instance with $N$ bits and $M$ clauses, where $M$ grows at most polynomially with $N$, the cost is
\begin{equation}
    C_{\mathrm{EC3}}(x)=\sum_{a=1}^{M}(x_{i_a}+x_{j_a}+x_{k_a}-1)^2.
\end{equation}
Choose $m$ pairwise disjoint clauses $I$ and define
\begin{equation}
    C_1(x)=\sum_{a\in I}(x_{i_a}+x_{j_a}+x_{k_a}-1)^2.
\end{equation}
The selected clauses define the feasible set
\begin{equation}
    F=\{x\in\bits^N:C_1(x)=0\},
\end{equation}
which has size
\begin{equation}
    L=|F|=3^m2^{N-3m}=2^N\left(\frac38\right)^m.
\end{equation}
Define the uniform reference state
\begin{equation}
    \ket{F}=\frac{1}{\sqrt{L}}\sum_{x\in F}\ket{x}.
\end{equation}
The learned-projector mixer generator $R_F=\Id-\ket{F}\!\bra{F}$ and the diagonal EC3 cost Hamiltonian both preserve the feasible subspace
\begin{equation}
    \calH_F=\operatorname{span}\{\ket{x}:x\in F\}.
\end{equation}
Let
\begin{equation}
    C_F=C_{\mathrm{EC3}}|_{\calH_F}
    =
    \sum_{x\in F}E_x\ket{x}\!\bra{x}.
\end{equation}
Each unselected clause contributes an integer between zero and four, so $E_x\in\mathbb Z_{\ge0}$ and $R=\max_{x\in F}E_x\le4(M-m)$.  Since $m\le N/3$, the feasible-set size obeys $L\ge3^{N/3}$.  Consequently,
\begin{equation}
    \frac{R^3}{\sqrt L}
    \le\frac{[4(M-m)]^3}{3^{N/6}}
    \longrightarrow0.
\end{equation}
These energy bounds and the asymptotic relation therefore follow from the EC3 construction.  We study the interpolation
\begin{equation}
    H_F(s)=(1-s)R_F+sC_F,\qquad 0\le s\le1.
\end{equation}

\begin{theorem}[Learned-projector gap on an EC3 feasible subspace]
\label{thm:ec3_gap}
Assume $C_F$ has a unique zero-cost state $x_\star\in F$, so $E_{x_\star}=0$ and $E_x\ge1$ for $x\in F\setminus\{x_\star\}$.  Define
\begin{equation}
    A_k=\frac{1}{L}\sum_{x\in F\setminus\{x_\star\}}E_x^{-k},
    \qquad k=1,2.
\end{equation}
The minimum gap between the lowest two eigenvalues of $H_F(s)$ is
\begin{equation}
    g_{\min}=
    \frac{2A_1}{1+A_1}\frac{1}{\sqrt{A_2L}}(1+o(1)).
    \label{eq:ec3_gap_formula}
\end{equation}
Equivalently, with
\begin{equation}
    q_N=\frac{2A_1}{(1+A_1)\sqrt{A_2}},
    \qquad
    \frac{L-1}{LR}\le q_N\le2,
    \label{eq:ec3_prefactor_bounds}
\end{equation}
the gap is
\begin{equation}
    g_{\min}
    =q_N\,2^{-N/2}\left(\frac83\right)^{m/2}(1+o(1)).
    \label{eq:ec3_gap_scaling}
\end{equation}
Since $R\le4(M-m)$ grows at most polynomially with $N$, $q_N$ lies between an inverse polynomial and a constant.
\end{theorem}

\begin{proof}
For $s>0$, set $a=(1-s)/s$ and write
\begin{equation}
    H_F(s)=(1-s)\Id+sK(a),
    \qquad
    K(a)=C_F-a\ket{F}\!\bra{F}.
\end{equation}
The scalar shift and positive factor $s$ preserve eigenvalue ordering.  If $\lambda_0(a)\le\lambda_1(a)\le\cdots$ are eigenvalues of $K(a)$, then the gap of $H_F(s)$ is $s(\lambda_1(a)-\lambda_0(a))$.

Let the distinct cost values of $C_F$ be $0=e_0<e_1<\cdots<e_q$, with spectral projectors $P_\alpha$, and define $w_\alpha=\bra{F}P_\alpha\ket{F}$.  Since the zero-cost state is unique, $w_0=1/L$.  The rank-one perturbation $-a\ket{F}\!\bra{F}$ leaves vectors in $P_\alpha\calH_F$ orthogonal to $P_\alpha\ket{F}$ at eigenvalue $e_\alpha$.  The collective eigenvalues satisfy
\begin{equation}
    \frac1a=\bra{F}(C_F-\lambda\Id)^{-1}\ket{F}
    =
    \sum_{\alpha=0}^{q}\frac{w_\alpha}{e_\alpha-\lambda}.
    \label{eq:secular_general}
\end{equation}
There is exactly one root $\lambda_-(a)<0$ and exactly one root $\lambda_+(a)\in(0,e_1)$.  These are the two lowest eigenvalues of $K(a)$.

For $|\lambda|\le1/2$, the identity
\begin{equation}
    \frac{1}{E_x-\lambda}
    =\frac{1}{E_x}+\frac{\lambda}{E_x^2}
    +\frac{\lambda^2}{E_x^2(E_x-\lambda)}
\end{equation}
gives the uniform expansion
\begin{equation}
    \frac1a=-\frac{1}{L\lambda}+A_1+A_2\lambda+\rho(\lambda),
    \qquad
    |\rho(\lambda)|\le2A_2\lambda^2.
    \label{eq:secular_expansion}
\end{equation}
Set $\eta=1/a-A_1$.  Multiplying by $\lambda$ gives
\begin{equation}
    A_2\lambda^2-\eta\lambda-\frac1L
    +O(A_2|\lambda|^3)=0.
    \label{eq:secular_quadratic}
\end{equation}
The quadratic part has roots
\begin{equation}
    \widetilde\lambda_\pm(\eta)=
    \frac{\eta\pm\sqrt{\eta^2+4A_2/L}}{2A_2}.
\end{equation}
In the window $|\eta|=O(\sqrt{A_2/L})$, these roots have size $O(1/\sqrt{A_2L})$.  Moreover,
\begin{equation}
    A_2\ge\frac{L-1}{LR^2}
    \quad\Longrightarrow\quad
    \frac{1}{\sqrt{A_2L}}
    \le\frac{R}{\sqrt{L-1}}=o(1).
\end{equation}
The remainder in Eq.~\eqref{eq:secular_quadratic} is therefore $o(1/L)$ at the quadratic roots.  Root perturbation gives, uniformly in this window,
\begin{equation}
    \lambda_+(\eta)-\lambda_-(\eta)=
    \frac{\sqrt{\eta^2+4A_2/L}}{A_2}
    (1+o(1)).
    \label{eq:collective_gap_window}
\end{equation}
The relation $a=1/(A_1+\eta)$ gives
\begin{equation}
    s=\frac{A_1+\eta}{1+A_1+\eta}.
\end{equation}
At $\eta=0$, Eqs.~\eqref{eq:collective_gap_window} and the expression for $s$ give a gap at most $2/\sqrt{L}\,(1+o(1))$.  Hence the global minimum tends to zero.

It remains to localize that global minimum.  Since $H_F(0)=R_F$ has gap one and
\begin{equation}
    \norm{H_F(s)-R_F}
    =s\norm{C_F-R_F}
    \le s(R+1),
\end{equation}
Weyl's inequality gives $g(s)\ge1/2$ whenever $s\le[4(R+1)]^{-1}$.  A global minimizer $s_{\min}$ therefore satisfies $s_{\min}>[4(R+1)]^{-1}$ for sufficiently large $L$.  At that point,
\begin{equation}
    \lambda_+(a)-\lambda_-(a)
    =\frac{g(s_{\min})}{s_{\min}}
    =O\!\left(\frac{R}{\sqrt{L}}\right)=o(1).
\end{equation}
Because $\lambda_-<0<\lambda_+$, both roots approach zero.  At these roots, the cubic remainder in Eq.~\eqref{eq:secular_quadratic}, relative to $1/L$, is $O(R^3/\sqrt{L})=o(1)$.  The quadratic relation therefore gives, at a global minimizer,
\begin{equation}
    g(s)=
    \frac{A_1+\eta}{1+A_1+\eta}
    \frac{\sqrt{\eta^2+4A_2/L}}{A_2}
    (1+o(1)).
    \label{eq:gap_global_quadratic}
\end{equation}
This expression localizes the minimizer.  For $\eta\ge0$, both factors preceding $1/A_2$ are minimized at $\eta=0$.  For $\eta<0$, put $t=A_1+\eta$.  Outside $|\eta|\le c\sqrt{A_2/L}$, the lower bound obtained by replacing the square root with $|\eta|$ is proportional to
\begin{equation}
    h(t)=\frac{t(A_1-t)}{1+t}.
\end{equation}
The derivative $h'(t)=(A_1-2t-t^2)/(1+t)^2$ changes sign once from positive to negative. Hence the minimum of $h$ on $[t_0,A_1-c\sqrt{A_2/L}]$ occurs at an endpoint, where $t_0=(4R+3)^{-1}$ is the boundary implied by Weyl's inequality.  At the upper endpoint, choosing any fixed $c>2$ makes Eq.~\eqref{eq:gap_global_quadratic} larger than the trial gap at $\eta=0$.  At the lower endpoint, $A_1\ge(L-1)/(LR)$ and $A_2\le1$ give a lower bound of order $R^{-2}$, whereas the trial gap is $O(L^{-1/2})$; the former is larger under $R^3/\sqrt{L}\to0$.  Hence every global minimizer satisfies $|\eta|=O(\sqrt{A_2/L})$, and Eq.~\eqref{eq:collective_gap_window} applies uniformly there.

The leading gap in this window is
\begin{equation}
    G(\eta)=
    \frac{A_1+\eta}{1+A_1+\eta}
    \frac{\sqrt{\eta^2+4A_2/L}}{A_2}.
\end{equation}
Differentiating $G$ shows that its minimizer satisfies
\begin{equation}
    \eta_{\min}=O\!\left(\frac{A_2}{A_1L}\right)
    =o\!\left(\sqrt{\frac{A_2}{L}}\right),
\end{equation}
where the final relation follows from $A_1\ge(L-1)/(LR)$ and $R/\sqrt{L}\to0$.  Substitution into $G$ proves Eq.~\eqref{eq:ec3_gap_formula}.

Finally, $A_1\ge(L-1)/(LR)$ and $A_2\le1$ give the lower bound in Eq.~\eqref{eq:ec3_prefactor_bounds}.  Cauchy--Schwarz gives $A_1^2\le[(L-1)/L]A_2$, which yields the upper bound.  Substituting $L=2^N(3/8)^m$ proves Eq.~\eqref{eq:ec3_gap_scaling}.
\end{proof}

\begin{remark}[Reference-state preparation]
The disjoint-clause construction gives $\ket{F}=\ket{W_3}^{\otimes m}\otimes\ket{+}^{\otimes(N-3m)}$, where $\ket{W_3}=(\ket{100}+\ket{010}+\ket{001})/\sqrt3$.  Preparing these factors in parallel realizes $\ket{F}$ exactly with a local constant-depth circuit.  This state can serve as the reference for learned-projector refinement; \suppref{app:stability} bounds the effect of an approximate preparation.
\end{remark}

\begin{remark}[Several satisfying assignments]
If $r>1$ satisfying assignments lie in $F$, the permutation-symmetric target sector replaces the weight $1/L$ by $r/L$.  Under analogous bounded-energy assumptions, its collective search gap scales as
\begin{equation}
    g_{\mathrm{search}}
    =q_{N,r}\sqrt{\frac{r}{L}}\,(1+o(1)),
\end{equation}
with a prefactor $q_{N,r}$ bounded below by an inverse polynomial.  This quantity is a search-sector gap because the zero-energy target space is degenerate.
\end{remark}

\subsection{Localization argument for block-XY refinement}
\label{app:gap_local}

We extend the localized-crossing analysis of Altshuler et al.\ to block-XY refinement on the feasible subspace \cite{altshuler2010anderson}.  For $s>0$, write the interpolation as $H_{\mathrm{XY}}(s)=s(C_F-\lambda B_F)$, where $\lambda=(1-s)/s$.  Here $B_F$ is the adjacency operator of the feasible-configuration graph $G_F$, with unit hopping amplitude along each edge.

Within each selected triple, block-XY mixing moves the single excitation among three positions; each unconstrained bit admits a single-bit flip.  The configuration graph is therefore the Cartesian product of $m$ copies of $K_3$ and an $(N-3m)$-dimensional hypercube.  If configurations $x,y\in F$ differ in $q_T$ selected triples and $q_U$ unconstrained bits, then
\begin{equation}
    \ell=d_{G_F}(x,y)=q_T+q_U,
    \qquad
    d_H(x,y)=2q_T+q_U.
\end{equation}
Thus $d_H(x,y)/2\le\ell\le d_H(x,y)$, and configurations separated by $\Theta(N)$ bit changes remain separated by $\Theta(N)$ allowed mixer moves.  Every term containing fewer than $\ell$ hopping operators has zero matrix element between these configurations.  Their effective tunnelling coupling therefore has no contribution below order $\ell$.

For the block-XY extension, we adopt the localized-crossing description on this graph.  Competing low-energy states remain localized around distant configurations and undergo an avoided crossing at $\lambda_c\le N^{-a}$ for some constant $a>0$.  We also use the path-sum estimate of the localization analysis: intermediate energy denominators compensate the factorial growth in the number of tunnelling paths.  This gives $|V_{xy}(\lambda_c)|\lesssim(A\lambda_c)^\ell$ with $A=O(1)$ \cite{altshuler2010anderson}.  Here the path length is the number of block exchanges and unconstrained-bit flips.  At the crossing, $s_c=(1+\lambda_c)^{-1}$, and the two-level gap is controlled by this coupling.  Consequently,
\begin{equation}
    g_{\min}^{\mathrm{XY}}
    \lesssim 2s_c|V_{xy}(\lambda_c)|
    \lesssim 2(AN^{-a})^\ell
    \le\exp[-cN\log N+O(N)],
    \qquad c>0.
\end{equation}
Block-XY mixing therefore retains the high-order tunnelling mechanism in this localized-crossing description.  In contrast, the learned-projector generator has $\bra{y}R_F\ket{x}=-1/L$ for every distinct pair $x,y\in F$.  Its collective coupling produces the minimum gap established in Theorem~\ref{thm:ec3_gap}, with $L^{-1/2}$ scaling up to a polynomial prefactor.

\section{Lie-algebraic analysis of QAOA trainability}
\label{app:trainability_theory}

This section summarizes published results connecting circuit generators to loss concentration and applies their notation to the QAOA mixer comparison.  Lie-algebraic approaches to trainability were developed by Larocca et al.\ \cite{larocca2022diagnosing}; the exact loss-variance formula used here is due to Ragone et al.\ \cite{ragone2024lie}.

\subsection{Dynamical Lie algebras and loss variance}
\label{app:dla_variance}

For Hermitian circuit generators $H_1,\ldots,H_K$, define the dynamical Lie algebra (DLA)
\begin{equation}
    \mathfrak g
    =\operatorname{Lie}_{\mathbb R}\{iH_1,\ldots,iH_K\}
    =\bigoplus_{\alpha=1}^{r}\mathfrak g_\alpha\oplus\mathfrak z,
\end{equation}
where the $\mathfrak g_\alpha$ are simple ideals and $\mathfrak z$ is the center.  Let $\mathcal G$ be the associated compact connected dynamical Lie group.  For an initial density operator $\rho$ and a Hermitian observable $O$, write
\begin{equation}
    \mathcal L(U)=\operatorname{Tr}(U\rho U^\dagger O).
\end{equation}
Let $A_\alpha$ and $A_{\mathfrak z}$ denote the Hilbert--Schmidt projections of a Hermitian operator $A$ onto $i\mathfrak g_\alpha$ and $i\mathfrak z$, respectively.

\begin{theorem}[Loss variance; Ragone et al., Theorem 1]
\label{thm:app_dla_variance}
Suppose $O\in i\mathfrak g$ or $\rho\in i\mathfrak g$.  For Haar-distributed $U\in\mathcal G$,
\begin{align}
    \mathbb E_{\mathcal G}[\mathcal L]
    &=\operatorname{Tr}(\rho_{\mathfrak z}O_{\mathfrak z}),\\
    \operatorname{Var}_{\mathcal G}[\mathcal L]
    &=\sum_{\alpha=1}^{r}
    \frac{\operatorname{Tr}(\rho_\alpha^2)
          \operatorname{Tr}(O_\alpha^2)}
         {\dim\mathfrak g_\alpha}.
\end{align}
These expressions also hold for circuit ensembles forming a unitary $2$-design over $\mathcal G$ \cite{ragone2024lie}.
\end{theorem}

A unitary $2$-design matches the first two unitary moments of Haar measure.  This supplies the statistical meaning of the sufficiently deep random-circuit limit.  The denominators quantify the dimensions of the accessible operator spaces; the numerators specify how the initial state and measured objective occupy those spaces.

Loss concentration concerns variation between parameter choices, rather than measurement variance within a fixed quantum state.  Chebyshev's inequality gives
\begin{equation}
    \Pr\!\left(
      |\mathcal L-\mathbb E\mathcal L|\ge\delta
    \right)
    \le\frac{\operatorname{Var}[\mathcal L]}{\delta^2},
    \qquad \delta>0.
\end{equation}
An exponentially small loss variance therefore suppresses resolvable differences across random parameter choices.  The connection between cost concentration and exponentially vanishing parameter gradients is developed by Arrasmith et al.\ \cite{arrasmith2022equivalence}.

\subsection{Variance results for transverse-field and Grover mixers}
\label{app:qaoa_mixer_variance}

The following statements use Haar averages over the respective dynamical Lie groups.

\paragraph{Transverse-field QAOA.}
For a graph $G=(V,E)$ on $N$ vertices, define
\begin{equation}
    B=\sum_{v\in V}X_v,
    \qquad
    C_G=\sum_{(u,v)\in E}w_{uv}Z_uZ_v,
    \qquad
    \mathfrak g_X=\operatorname{Lie}_{\mathbb R}\{iB,iC_G\}.
\end{equation}
Minimizing $C_G$ is equivalent to maximizing the weighted cut.  Following Mao et al., take
\begin{equation}
    \rho_+=(\ket{+}\!\bra{+})^{\otimes N},
    \qquad
    O_G=\frac{C_G}{\sqrt{\sum_{(u,v)\in E}w_{uv}^2}}.
\end{equation}
Write $\mathcal G_X$ for the corresponding dynamical Lie group.

\begin{theorem}[QAOA--MaxCut; Mao et al., Theorems 2 and 3]
\label{thm:app_x_mixer_variance}
For an unweighted Erd\H{o}s--R\'enyi graph $G\sim G(N,1/2)$, the following holds with probability at least $1-\exp[-\Omega(N)]$:
\begin{equation}
    \mathfrak g_X=\bigoplus_{\alpha=1}^{r}\mathfrak g_\alpha,
    \qquad
    r\in\{1,2\},
    \qquad
    \dim\mathfrak g_\alpha=\Theta(4^N),
\end{equation}
and
\begin{equation}
    \operatorname{Var}_{\mathcal G_X}
    \!\left[
      \operatorname{Tr}(U\rho_+U^\dagger O_G)
    \right]
    =O(2^{-N}).
\end{equation}
The same conclusions hold almost surely for connected graphs other than paths and cycles with independently drawn continuous edge weights \cite{mao2025qaoa}.
\end{theorem}

\paragraph{Grover-mixer QAOA.}
Let $C=\sum_a\lambda_a P_a$ be a diagonal cost Hamiltonian, with distinct eigenvalues $\lambda_a$ and spectral projectors $P_a$.  For a pure reference state $\ket{\xi}$, define
\begin{equation}
    \Pi_\xi=\ket{\xi}\!\bra{\xi},
    \qquad
    \mathcal K=\operatorname{span}\{P_a\ket{\xi}:P_a\ket{\xi}\ne0\},
    \qquad
    d=\dim\mathcal K\ge2.
\end{equation}
Thus $d$ counts distinct cost values represented in the reference state.  Relabel the supported eigenvalues as $\lambda_1,\ldots,\lambda_d$.

\begin{theorem}[Grover-mixer algebra and variance; Tsvelikhovskiy et al.]
\label{thm:app_grover_variance}
For $\mathfrak g_\xi=\operatorname{Lie}_{\mathbb R}\{i\Pi_\xi,iC\}$, the restriction to $\mathcal K$ is $\mathfrak u(d)$, and the simple component is $\mathfrak{su}(d)$, of dimension $d^2-1$.  The full algebra has the decomposition
\begin{equation}
    \mathfrak g_\xi
    \cong
    \mathfrak u(d)\oplus
    \operatorname{span}_{\mathbb R}\{iC|_{\mathcal K^\perp}\}.
\end{equation}
Its dimension is $d^2+1$ when $C|_{\mathcal K^\perp}\ne0$, and $d^2$ otherwise.  Let $\mathcal G_\xi$ be the corresponding dynamical Lie group.  For $\rho_\xi=\Pi_\xi$,
\begin{equation}
    \operatorname{Var}_{\mathcal G_\xi}
    \!\left[
      \operatorname{Tr}(U\rho_\xi U^\dagger C)
    \right]
    =
    \frac{1}{d(d+1)}
    \sum_{a=1}^{d}(\lambda_a-\bar\lambda)^2,
    \qquad
    \bar\lambda=\frac1d\sum_{a=1}^{d}\lambda_a.
\end{equation}
Here each supported eigenvalue is counted once, irrespective of its degeneracy \cite{tsvelikhovskiy2025provable}.
\end{theorem}

When $C$ encodes an integer-valued objective of the form
\begin{equation}
    f(x)=\sum_{a=1}^{T}f_a(x_{S_a}),
    \qquad
    |S_a|=s,
    \qquad
    T\le\binom Ns,
    \qquad
    |f_a|\le M_{\mathrm{loc}},
\end{equation}
with $C\ket{x}=f(x)\ket{x}$, Theorem IV.2 of the same work gives
\begin{equation}
    \operatorname{Var}_{\mathcal G_\xi}
    \!\left[
      \operatorname{Tr}(U\rho_\xi U^\dagger\widehat C)
    \right]
    \ge\frac{(s!)^2}{12M_{\mathrm{loc}}^2N^{2s}},
    \qquad
    \widehat C=\frac{C}{\norm C},
\end{equation}
where $\norm C$ is the operator norm.  For fixed $s$ and polynomially bounded $M_{\mathrm{loc}}$, this is an inverse-polynomial lower bound.  The unweighted MaxCut objective counting cut edges has $s=2$ and $M_{\mathrm{loc}}=1$, giving $1/(3N^4)$.

The complementary projector used by LP-QAOA satisfies
\begin{equation}
    e^{-i\beta(\Id-\Pi_\xi)}
    =e^{-i\beta}e^{i\beta\Pi_\xi}.
\end{equation}
It therefore produces the same expectation-value landscape after reversing the mixer angle; the global phase has no effect.

\section{Resource accounting and numerical methods}
\label{app:resource}

\subsection{Full state-vector simulation}

All reported probabilities were computed from full state vectors after classical angle optimization.  BCST simulations used a complete product basis containing only states that satisfy each block-cardinality constraint; this basis contains every state reached by the corresponding block-XY mixer.  SBM simulations used complete binary state vectors over the selected spin variables.

Diagonal phase separators were applied pointwise, and fixed transverse-field and XY mixers were applied through their exact actions in the chosen basis.  The LP mixer was applied as
\begin{equation}
    e^{-i\beta R_\Phi}\ket{\psi}
    =
    e^{-i\beta}\ket{\psi}
    +(1-e^{-i\beta})\ket{\Phi}\braket{\Phi}{\psi},
\end{equation}
which reproduces the ideal circuit action in Eq.~\eqref{eq:history_circuit}.

\subsection{Resource accounting and metrics}

This section specifies resource units (RU), repetitions to solution ($\mathrm{RTS}_{99}$) for a 99\% target-observation probability, and the total-cost calculation.  Because $\mathrm{RTS}_{99}$ is nonlinear in success probability, the aggregation convention is stated for each comparison below.

Let $P_0\ket{0^N}=\ket{\Phi_0}$ and $r_0=\RU(P_0)$.  We use $\mathrm{MCP}_N(\beta)$ for the $N$-qubit parameterized multi-controlled phase rotation that implements the selective phase in Eq.~\eqref{eq:history_circuit}, and write $\RU(\mathrm{MCP}_N)$ for its angle-independent resource proxy.  Define the complete preparation through stage $j$ as
\[
    A_j=U_j\cdots U_1P_0,
    \qquad a_j=\RU(A_j),
    \qquad a_0=r_0.
\]
If $A_{j-1}$ prepares $\ket{\Phi_{j-1}}$, then
\begin{equation}
    e^{-i\beta R_{\Phi_{j-1}}}
    =
    A_{j-1}e^{-i\beta(\Id-\ket{0^N}\!\bra{0^N})}A_{j-1}^{\dagger}.
\end{equation}
The corresponding LP-mixer layer has logical resource cost
\begin{equation}
    \RU(R_{\Phi_{j-1}})=
    2a_{j-1}+\RU(\mathrm{MCP}_N).
    \label{eq:history_cost}
\end{equation}
The complete preparation cost obeys
\begin{equation}
    a_j=
    \begin{cases}
    a_{j-1}+p_j\{\RU(C_j)+\RU(M_j)\}, & \text{fixed mixer},\\
    (1+2p_j)a_{j-1}+p_j\{\RU(C_j)+\RU(\mathrm{MCP}_N)\}, & \text{LP mixer}.
    \end{cases}.
    \label{eq:resource_recurrence}
\end{equation}
Let $\mathcal T$ be the prescribed target set and let $V_{\mathcal T}$ denote the terminal evaluation used to determine whether a measured configuration belongs to $\mathcal T$.  The reported per-shot cost is
\begin{equation}
    \RU_{\mathcal T}(A_j)=a_j+\tau_{\mathcal T},
    \qquad
    \tau_{\mathcal T}=\RU(V_{\mathcal T}).
    \label{eq:reported_ru}
\end{equation}
This terminal evaluation is applied once after state preparation and is not part of any replayed circuit.
If $\mathcal H$ denotes the set of stages that use an LP mixer, including $P_0$ changes the per-shot logical cost by
\begin{equation}
    \Delta\RU_{\mathrm{LP}}
    =r_0\prod_{j\in\mathcal H}(1+2p_j).
    \label{eq:initial_state_sensitivity}
\end{equation}
For a two-stage circuit with one LP-mixer stage, Eq.~\eqref{eq:initial_state_sensitivity} reduces to $(2p_2+1)r_0$, whereas a direct circuit receives the one-time correction $r_0$.  The multilevel BCST circuit contains two successive LP-mixer stages and is evaluated with the full recurrence in Eq.~\eqref{eq:resource_recurrence}; its explicit resource expressions are given in \suppref{app:bcst_multilevel_resources}.  The SBM calculations assign zero cost to the initial $\ket{+}^{\otimes N}$ preparation.  The BCST calculations assign one unit to each of the five initial block-state preparations and include this cost in every replay.

For the additional two-level BCST experiments, let $P_{\mathrm{blk}}$, $C$, $X$ and $O$ denote the resource costs of the initial block-state preparation, conflict phase separator, block-XY mixer and final-objective phase separator, respectively, and let $S=\RU(\mathrm{MCP}_N)$.  The RU model assigns $\RU(\mathrm{MCP}_N)=N^2$.  The fixed resource counts are $(P_{\mathrm{blk}},C,X,O,S)=(5,25,50,3{,}458,625)$ at $N=25$ and $(5,30,75,12{,}110,900)$ at $N=30$.  If $m_4$ and $m_6$ are the numbers of degree-four and degree-six objective terms, then $O=14m_4+22m_6$, because these degree-$k$ diagonal terms are assigned $4k-2$ RU\@.  A block-XY mixer on block sizes $(n_1,\ldots,n_m)$ is assigned
\begin{equation}
    \RU(M_{\mathrm{XY}})=\sum_{b=1}^{m}\binom{n_b}{2},
\end{equation}
one unit for each pairwise XY generator.  With $P_1=P_{\mathrm{blk}}+p_1(C+X)$ and $\tau_{\mathcal T}=O$, the per-shot costs used in the BCST tables are
\begin{align}
    \RU_{\mathrm{Stage\text{-}1}}&=P_1+O,\\
    \RU_{\mathrm{LP}}(p)&=P_1+p(O+2P_1+S)+O,\\
    \RU_{\mathrm{warm\text{-}XY}}(p)&=P_1+p(C+O+X)+O,\\
    \RU_{\mathrm{block\text{-}XY}}(p)&=P_{\mathrm{blk}}+p(C+O+X)+O.
    \label{eq:bcst_current_resource_counts}
\end{align}
The block-XY (combined phase) and block-XY (separate phases) baselines use the same resource expression and differ only in their phase-angle parameterization.  Resources for the controlled mixer, reference-state and objective ablations follow the substitutions listed in \suppref{app:bcst_causal}.

For $N=25$, $O=225(14)+14(22)=3{,}458$ and $P_1=5+12(25+50)=905$, so Eq.~\eqref{eq:bcst_current_resource_counts} becomes
\begin{align*}
    \RU_{\mathrm{Stage\text{-}1}}&=4{,}363,&
    \RU_{\mathrm{LP}}(p)&=4{,}363+5{,}893p,\\
    \RU_{\mathrm{warm\text{-}XY}}(p)&=4{,}363+3{,}533p,&
    \RU_{\mathrm{block\text{-}XY}}(p)&=3{,}463+3{,}533p.
\end{align*}
For $N=30$, $O=788(14)+49(22)=12{,}110$ and $P_1=5+12(30+75)=1{,}265$, giving
\begin{align*}
    \RU_{\mathrm{Stage\text{-}1}}&=13{,}375,&
    \RU_{\mathrm{LP}}(p)&=13{,}375+15{,}540p,\\
    \RU_{\mathrm{warm\text{-}XY}}(p)&=13{,}375+12{,}215p,&
    \RU_{\mathrm{block\text{-}XY}}(p)&=12{,}115+12{,}215p.
\end{align*}

For SBM, the logical RU model assigns one unit to each two-body diagonal term and one unit to each transverse-field mixer term in a standard QAOA layer.  The terminal evaluation has cost $\tau_{\mathrm{SBM}}=\RU(C_{\mathrm{fine}})=\binom{N}{2}$.  Thus, a standard circuit using either the fine Hamiltonian or the LP stage-1 objective at depth $p$ has
\begin{equation}
    \RU_{\mathrm{direct}}(p)=p\{\RU(C)+N\}+\tau_{\mathrm{SBM}},
\end{equation}
where $C$ is the coarse or fine pairwise Hamiltonian.  In the resource accounting used here, $r_0=0$. A soft LP-QAOA circuit with depths $(p_1,p_2)$ includes the first-stage circuit once and each final learned-projector layer from Eq.~\eqref{eq:history_cost} together with the final diagonal phase:
\begin{equation}
    \RU_{\mathrm{LP}}
    =
    \RU(U_1)+p_2\{\RU(C_{\mathrm{fine}})+2\RU(U_1)+N^2\}+\tau_{\mathrm{SBM}}.
    \label{eq:soft_ru}
\end{equation}
Here $\RU(U_1)=p_1\{\RU(C_{\mathrm{coarse}})+N\}$.  A warm-start circuit has $\RU_{\mathrm{warm}}=\RU(U_1)+p_2\{\RU(C_{\mathrm{fine}})+N\}+\tau_{\mathrm{SBM}}$.  At $N=24$, both pairwise Hamiltonians contain $\binom{24}{2}=276$ terms.  The resulting per-shot costs are $1{,}476$ RU for LP stage-1 objective QAOA and standard QAOA at $p=4$, $2{,}076$ and $2{,}676$ RU for standard QAOA at $p=6$ and $p=8$, $1{,}776$ RU for warm-start QAOA at $(4,1)$, and $4{,}728$ RU for LP-QAOA at $(4,1)$.

For a final state $\ket{\psi}$ and a prescribed target set $\mathcal T$,
\begin{equation}
    P_{\mathrm{succ}}(\mathcal T)=\sum_{z\in\mathcal T}|\braket{z}{\psi}|^2.
\end{equation}
Setting the required target-observation probability to $P_{\mathrm{target}}=0.99$ gives
\begin{equation}
    \mathrm{RTS}_{99}=
    \begin{cases}
    \infty, & P_{\mathrm{succ}}=0,\\
    \left\lceil
    \dfrac{\log(1-P_{\mathrm{target}})}
          {\log(1-P_{\mathrm{succ}})}
    \right\rceil, & 0<P_{\mathrm{succ}}<1,\\
    1, & P_{\mathrm{succ}}=1,
    \end{cases}
    \qquad
    \mathrm{cost}=\mathrm{RTS}_{99}\times\mathrm{RU}.
\end{equation}

For the SBM comparison, we summarize repetition cost with the continuous relaxation
\begin{equation}
    \mathcal C_{99}^{\mathrm{cont}}
    =
    \RU\max\!\left\{1,\frac{\log(0.01)}{\log(1-P_{\mathrm{succ}})}\right\},
    \label{eq:sbm_continuous_cost}
\end{equation}
which evaluates the 99\%-success repetition formula before integer rounding. Its endpoint values are defined by continuity: $+\infty$ at $P_{\mathrm{succ}}=0$ and $\RU$ at $P_{\mathrm{succ}}=1$.

\subsection{Aggregation and comparison rules}

All comparisons within a final-objective realization or problem instance used the same enumerated target set.  The primary BCST outcome is unique-optimum total cost, $\mathrm{RTS}_{99}\times\mathrm{RU}$.  In the multilevel BCST study, the four optimizer starts were first aggregated by geometric mean within each final-objective realization; geometric means and bootstrap intervals were then calculated over the ten realization-level values.  In the additional two-level BCST experiments, for each method and instance we selected the lowest-energy circuit across the stated restarts. We then calculated target probabilities, $\mathrm{RTS}_{99}$ and total costs and summarized each outcome over the stated instance set by a geometric mean.  Outcomes for the sets containing the two and eight lowest-cost states are sensitivity analyses for this additional benchmark.  For each SBM parameter setting, nine graph instances were summarized by medians and paired win counts; the instance with disjoint coarse and fine ground-state sets was analysed separately as a coarse--fine mismatch control.  Source Data provide final-objective realizations, graph definitions, optimization settings and individual outcomes.

\section{Multilevel BCST benchmark and numerical experiments}
\label{app:bcst_multilevel}

\subsection{Benchmark construction}

For a general BCST instance, let $V$ denote the sites, let $\mathcal L$ denote the labels and let $k_v$ be the number of labels selected at site $v$.  A binary variable $x_{vc}$ records whether label $c$ is selected at site $v$, and the block-cardinality search space is
\[
    \mathcal X_{\mathbf k}
    =\left\{x\in\bits^{|V|\times|\mathcal L|}:\sum_{c\in\mathcal L}x_{vc}=k_v\ \text{for every }v\in V\right\}.
\]
The multilevel instances studied here use five sites $v\in\{0,\ldots,4\}$, six channel labels $c\in\{0,\ldots,5\}$ and $k_v=2$ for every site.  Writing this search space as $\mathcal X_2$ gives
\begin{equation}
    |\mathcal X_2|=\binom{6}{2}^{5}=759{,}375.
    \label{eq:bcst_multilevel_space}
\end{equation}
The five sites form the cycle $C_5$.  Neighbouring sites may not select the same channel, and the network-wide assignment counts must equal $q=(2,2,2,2,1,1)$.  These two constraints are represented by
\begin{align}
    C(x)&=\sum_{\{u,v\}\in E(C_5)}\sum_{c=0}^{5}x_{uc}x_{vc},
    \label{eq:bcst_multilevel_conflict}\\
    Q(x)&=\sum_{c=0}^{5}\left(\sum_{v=0}^{4}x_{vc}-q_c\right)^2.
    \label{eq:bcst_multilevel_quota}
\end{align}
Within $\mathcal X_2$, the zero-energy spaces of $Q$ and $C$ contain 4,530 and 6,570 configurations, respectively, and their intersection
\begin{equation}
    \mathcal F=\{x\in\mathcal X_2:Q(x)=C(x)=0\}
    \label{eq:bcst_multilevel_feasible_set}
\end{equation}
contains 390 configurations.

The final objective combines a site--channel assignment cost with a penalty for third-order intermodulation (IM3) products whose frequencies coincide with assigned channels.  For each channel $c$, define the pairs of source channels that produce an IM3 product at $c$,
\begin{align*}
\mathcal G_0&=\{(1,2),(2,4)\}, &
\mathcal G_1&=\{(2,3),(3,5)\}, &
\mathcal G_2&=\{(0,1),(3,4)\},\\
\mathcal G_3&=\{(1,2),(4,5)\}, &
\mathcal G_4&=\{(0,2),(2,3)\}, &
\mathcal G_5&=\{(1,3),(3,4)\},
\end{align*}
and let
\begin{align}
    L_{vc}(x)
    &=\sum_{u\in V\setminus\{v\}}\sum_{(a,b)\in\mathcal G_c}x_{ua}x_{ub},
    \label{eq:bcst_multilevel_im3_count}\\
    O_s(x)
    &=256\sum_{v=0}^{4}\sum_{c=0}^{5}x_{vc}L_{vc}(x)^2
      +\sum_{v=0}^{4}\sum_{c=0}^{5}A^{(s)}_{vc}x_{vc}.
    \label{eq:bcst_multilevel_objective}
\end{align}
Here $A^{(s)}$ is the site--channel assignment-cost matrix for final-objective realization $s$.

The binary occupations have a direct frequency-assignment interpretation (Fig.~\ref{fig:bcst_frequency_assignment}).  A site $v$ represents a radio site, and $x_{vc}=1$ assigns channel $c$ to it; the block-cardinality condition fixes the number of channels assigned to each site.  The conflict cost $C$ excludes reuse of a channel at adjacent sites, whereas $Q$ enforces the prescribed network-wide number of assignments of each channel.  Among assignments satisfying both rules, $A^{(s)}_{vc}x_{vc}$ gives the cost of assigning channel $c$ to site $v$.  If channel $c$ is assigned to site $v$, $L_{vc}$ counts prescribed channel pairs at other sites that generate IM3 products whose frequencies coincide with channel $c$.  The term $x_{vc}L_{vc}^{2}$ assigns a quadratic penalty to this count.  The construction therefore separates the per-site channel counts, adjacent-site exclusion and network-wide assignment quotas from the final objective, which combines assignment costs with penalties for intermodulation interference.

\begin{figure}[!htbp]
\centering
\includegraphics[width=0.98\linewidth]{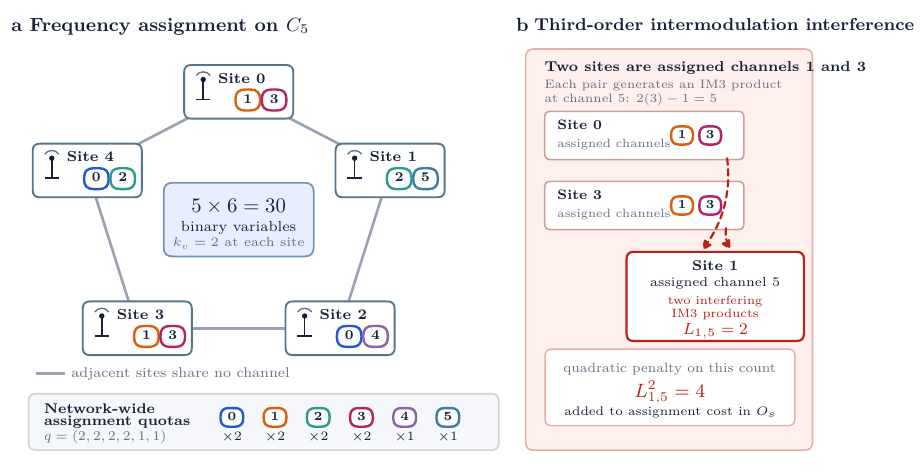}
\caption{\textbf{Frequency-assignment interpretation of BCST.} \textbf{a}, Five sites are assigned two of six channels subject to adjacent-site exclusion and network-wide assignment quotas. \textbf{b}, The channel pair $\{1,3\}$ assigned at each of Sites 0 and 3 generates a third-order intermodulation (IM3) product at channel 5, which is assigned to Site 1; hence $L_{1,5}=2$ and $L_{1,5}^{2}=4$.}
\label{fig:bcst_frequency_assignment}
\end{figure}
\FloatBarrier

After using $x_{vc}^2=x_{vc}$ and collecting repeated monomials, $O_s$ contains 30, 240 and 720 distinct monomial supports of degrees 1, 3 and 5, respectively. The corresponding diagonal Hamiltonian therefore has locality at most five.  The ten reported $5\times6$ matrices were drawn independently from the integers 1 through 20 using PCG64 seeds 2026091100, 2026091101 and 2026091103--2026091110.  Each produced a unique minimum of $O_s$ over $\mathcal F$.

\subsection{Circuit families and optimization protocol}

Let $\ket{\Phi_0}$ be the uniform superposition over $\mathcal X_2$, and let
\begin{equation}
    M_{\mathrm{XY}}
    =\frac{1}{2}\sum_{v=0}^{4}\sum_{0\le a<b\le5}
    \left(X_{va}X_{vb}+Y_{va}Y_{vb}\right)
    \label{eq:bcst_multilevel_xy_mixer}
\end{equation}
be the complete within-site XY mixer.  For the numerical optimization, we used the scaled diagonals
\begin{equation}
    \widetilde Q=Q/48,
    \qquad
    \widetilde C=C/10,
    \qquad
    \widetilde O_s=O_s/2048,
    \qquad
    D_s=\frac{O_s+1024(C+Q)}{2048}.
    \label{eq:bcst_multilevel_scaled_costs}
\end{equation}
Stage 1 used $M_{\mathrm{XY}}$ and $\widetilde C$ at depth 12, and its optimized circuit was frozen to prepare $\ket{\Phi_1}$.  Stage 2 used the learned-projector mixer $R_{\Phi_1}$ and the unscaled global-occupation penalty $Q$ as both its phase separator and expected-energy loss:
\begin{equation}
    \ket{\psi_2(\theta_2)}=
    \prod_{\ell=14}^{1}
    e^{-i\beta_{2,\ell}R_{\Phi_1}}e^{-i\gamma_{2,\ell}Q}\ket{\Phi_1},
    \qquad
    \mathcal L_2=\bra{\psi_2}Q\ket{\psi_2}.
    \label{eq:bcst_multilevel_feasibility_separator}
\end{equation}
The resulting depth-14 circuit was frozen to prepare $\ket{\Phi_2}$.  At objective-refinement depth $p$, LP-QAOA then prepared
\begin{equation}
    \ket{\psi_s(\theta;p)}=
    \prod_{\ell=p}^{1}
    e^{-i\beta_{\ell}R_{\Phi_2}}
    e^{-i\gamma_{\ell}\widetilde O_s}
    \ket{\Phi_2},
    \label{eq:bcst_multilevel_final_circuit}
\end{equation}
and selected the checkpoint with the lowest expectation value of $D_s$.  Thus, the final phase separator contains only the weighted objective, whereas the classical loss also penalizes violations of the adjacency and channel-assignment constraints.  Because $C$, $Q$ and $\mathcal F$ do not depend on $A^{(s)}$, the first two frozen preparations were shared across all ten final-objective realizations.

Table~\ref{tab:bcst_multilevel_families} gives the circuit families used in the main comparison and ablation tests.  Separate phases assign an independent angle to each listed diagonal within a layer.  Two-stage LP-QAOA uses the unscaled $Q$, either separately from $\widetilde O_s$ or in the combined phase $Q+\widetilde O_s$; combined-phase block-XY circuits use $D_s$.  All final-objective optimization uses the expected-energy loss $\langle D_s\rangle$.

\begin{table}[!htbp]
\centering
\caption{Circuit families in the multilevel BCST study.}
\label{tab:bcst_multilevel_families}
\small
\begin{tabularx}{\textwidth}{@{}>{\raggedright\arraybackslash}p{0.22\textwidth}>{\raggedright\arraybackslash}p{0.23\textwidth}>{\raggedright\arraybackslash}X>{\raggedright\arraybackslash}p{0.18\textwidth}@{}}
\toprule
Construction & Input state & Active phase separators & Active mixer \\
\midrule
LP-QAOA & $\ket{\Phi_2}$ & $\widetilde O_s$ & $R_{\Phi_2}$ \\
LP-QAOA (two stages, separate phases) & $\ket{\Phi_1}$ & $Q$ and $\widetilde O_s$, separate & $R_{\Phi_1}$ \\
LP-QAOA (two stages, combined phase) & $\ket{\Phi_1}$ & $Q+\widetilde O_s$, combined & $R_{\Phi_1}$ \\
Block-XY QAOA (separate phases) & $\ket{\Phi_0}$ & $\widetilde Q$, $\widetilde C$ and $\widetilde O_s$, separate & $M_{\mathrm{XY}}$ \\
Block-XY QAOA (combined phase) & $\ket{\Phi_0}$ & $D_s$, combined & $M_{\mathrm{XY}}$ \\
\addlinespace[3pt]
No joint-feasibility stage & $\ket{\Phi_1}$ & $\widetilde O_s$ & $R_{\Phi_1}$ \\
Warm-start block-XY & $\ket{\Phi_1}$ & $D_s$, combined & $M_{\mathrm{XY}}$ \\
Uniform-state projector & $\ket{\Phi_0}$ & $\widetilde Q$, $\widetilde C$ and $\widetilde O_s$, separate & $R_{\Phi_0}$ \\
\bottomrule
\end{tabularx}
\end{table}

The frozen depth-12 conflict preparation was obtained with Adam at a learning rate of 0.035 and 4,800 objective evaluations per optimizer start, including the initial evaluation.  The four preparations with seeds 2026081100--2026081103 were paired with the subsequent optimizer starts.  Stage 2 used 800 objective evaluations per start, including the initial evaluation, with uniform $[-\pi,\pi]$ initial angles and checkpoint selection by minimum expected $Q$.  Every nonzero-depth refinement circuit used Adam with a learning rate of 0.035, 2,400 objective evaluations (2,399 updates) and single-precision complex state vectors.  All three LP-QAOA variants used refinement depths $p\in\{1,2,4,6,8,10,12,14,16,24,32,40,48\}$ on each of the ten realizations; three-stage LP-QAOA also included $p=11$.  The block-XY baselines and ablation controls used $p\in\{1,2,4,8,16,32,64\}$, with additional $p=6,10,12,14$ for separate-phase block-XY.  Every nonzero-depth point was optimized independently, without angle transfer between depths.  For each non-warm-start family, one uniform $[-\pi,\pi]$ angle array was generated through its maximum depth and truncated for smaller depths.  Warm-start block-XY used deterministic $\pm10^{-3}$ initial angles.  The $p=0$ point is the corresponding unoptimized input state, with its preparation and terminal objective-evaluation costs included.

The four optimizer starts are repeated searches within a final-objective realization.  We first took their geometric mean for each realization and then calculated geometric means over realizations.  The intervals in Fig.~\ref{fig:bcst_main_comparison}a,b and Fig.~\ref{fig:bcst_component_tests} were obtained from 20,000 nonparametric bootstrap resamples of the ten realization-level values.  Each method's nonzero depth was selected by minimum geometric-mean total cost over four starts in the depth sweep and then used for all ten realizations.  Both the depth curves and the selected-depth summaries use all ten realizations.

\subsection{Logical resource expressions}
\label{app:bcst_multilevel_resources}

The multilevel BCST resource model assigns $P=5$ RU to the block-cardinality preparation, $X=75$ RU to the complete block-XY mixer and $S=30^2=900$ RU to a selective phase.  Each quadratic occupation-product phase is assigned one RU, including all terms in $C$ and $Q$, giving $C=30$ RU\@.  Unary and higher-order occupation-product phases retain the assignment $4k-2$ RU for degree $k\ne2$.  Expanding $Q$ gives 30 unary and 60 quadratic supports, hence
\begin{equation}
    \RU(Q)=30(2)+60(1)=120.
    \label{eq:bcst_multilevel_quota_ru}
\end{equation}
The conflict supports are contained in the quadratic support set of $Q$, so the resource cost assigned to the union of the supports of $Q$ and $C$ is also 120 RU\@.  The final objective costs
\begin{equation}
    O=30(2)+240(10)+720(18)=15{,}420\ \mathrm{RU}.
    \label{eq:bcst_multilevel_objective_ru}
\end{equation}
Combining the constraint and objective supports gives
\begin{equation}
    J_{QCO}=30(2)+60(1)+240(10)+720(18)=15{,}480\ \mathrm{RU},
    \label{eq:bcst_multilevel_union_ru}
\end{equation}
whereas separate $Q$, $C$ and $O_s$ phases cost $120+30+15{,}420=15{,}570$ RU per layer and separate $Q$ and $O_s$ phases cost $120+15{,}420=15{,}540$ RU\@.  A combined $Q+\widetilde O_s$ phase also costs $J_{QCO}$: the 30 unary supports shared by $Q$ and $O_s$ are counted once.  The frozen preparation costs after the conflict and global-occupation stages are
\begin{align}
    P_C&=P+12(C+X)=1{,}265,\label{eq:bcst_multilevel_pc}\\
    P_Q&=(1+2\times14)P_C+14\{\RU(Q)+S\}=50{,}965.\label{eq:bcst_multilevel_pf}
\end{align}
Thus, the per-shot LP-QAOA resource cost is
\begin{equation}
    \RU_{\mathrm{LP}}(p)
    =(1+2p)P_Q+p(O+S)+O
    =66{,}385+118{,}250p.
    \label{eq:bcst_multilevel_lp_ru}
\end{equation}
In particular, $\RU_{\mathrm{LP}}(11)=1{,}367{,}135$.  A synthesis-cost sensitivity calculation assigning every quadratic occupation-product phase six RU sets $C=180$ and $\RU(Q)=420$, giving $\RU_{\mathrm{LP}}(p)=122{,}785+231{,}050p$, including $2{,}664{,}335$ RU at $p=11$.

The remaining multilevel expressions follow from the same resource accounting:
\begin{align*}
\RU_{\mathrm{two\text{-}stage,sep}}(p)
    &=P_C+p\{2P_C+\RU(Q)+O+S\}+O
      =16{,}685+18{,}970p,\\
\RU_{\mathrm{two\text{-}stage,comb}}(p)
    &=P_C+p(2P_C+J_{QCO}+S)+O
      =16{,}685+18{,}910p,\\
\RU_{\mathrm{direct,sep}}(p)
    &=P+p\{\RU(Q)+C+O+X\}+O
      =15{,}425+15{,}645p,\\
\RU_{\mathrm{direct,comb}}(p)
    &=P+p(J_{QCO}+X)+O
      =15{,}425+15{,}555p,\\
\RU_{\mathrm{no\text{-}feasibility}}(p)
    &=P_C+p(2P_C+O+S)+O
      =16{,}685+18{,}850p,\\
\RU_{\mathrm{warm\text{-}XY}}(p)
    &=P_C+p(J_{QCO}+X)+O
      =16{,}685+15{,}555p,\\
\RU_{\mathrm{uniform\text{-}projector}}(p)
    &=P+p\{\RU(Q)+C+O+2P+S\}+O
      =15{,}425+16{,}480p.
\end{align*}
In the reduced $N=18$ standard-QAOA diagnostic, the objective contains 18 unary, 72 cubic, 24 quartic and 84 quintic supports, so its terminal cost is
\begin{equation*}
    O_{18}=18(2)+72(10)+24(14)+84(18)=2{,}604.
\end{equation*}
The penalty-encoded phase adds 63 quadratic supports and therefore costs $2{,}604+63(1)=2{,}667$ RU\@.  Charging 18 RU each for the initial preparation and transverse-field mixer yields
\begin{equation*}
    \RU_{\mathrm{standard},18}(p)=18+p(2{,}667+18)+2{,}604
    =2{,}622+2{,}685p.
\end{equation*}

\begin{table}[!htbp]
\centering
\caption{Per-shot logical resource expressions for the multilevel BCST circuits.}
\label{tab:bcst_multilevel_ru}
\small
\begin{tabular}{@{}lp{0.34\textwidth}@{}}
\toprule
Construction & Resource expression \\
\midrule
LP-QAOA & $66{,}385+118{,}250p$ \\
LP-QAOA (two stages, separate phases) & $16{,}685+18{,}970p$ \\
LP-QAOA (two stages, combined phase) & $16{,}685+18{,}910p$ \\
Block-XY QAOA (separate phases) & $15{,}425+15{,}645p$ \\
Block-XY QAOA (combined phase) & $15{,}425+15{,}555p$ \\
No joint-feasibility stage & $16{,}685+18{,}850p$ \\
Warm-start block-XY & $16{,}685+15{,}555p$ \\
Uniform-state projector & $15{,}425+16{,}480p$ \\
Reduced $N=18$ standard QAOA diagnostic & $2{,}622+2{,}685p$ \\
\bottomrule
\end{tabular}
\end{table}

All expressions include one terminal evaluation of $O_s$.  They quantify circuit execution after angle optimization; the reported $\mathrm{RTS}_{99}\times\RU$ values combine these per-shot costs with the exact final-state success probabilities.

\subsection{Numerical summaries and additional objective realizations}

Table~\ref{tab:bcst_multilevel_selected_depths} gives the numerical values underlying the selected-depth comparisons in the main text.  Each construction retains its previously selected nonzero depth for all ten realizations.

\begin{table}[!htbp]
\centering
\caption{Multilevel BCST results at each construction's selected nonzero depth.}
\label{tab:bcst_multilevel_selected_depths}
\scriptsize
\resizebox{\textwidth}{!}{%
\begin{tabular}{@{}lrrrrc@{}}
\toprule
Construction & $p$ & Per-shot RU & $P_{\mathrm{succ}}$ & $\mathrm{RTS}_{99}\times\RU$ & LP-QAOA paired wins (cost, probability) \\
\midrule
LP-QAOA & 11 & 1,367,135 & $7.983\times10^{-2}$ & $7.638\times10^{7}$ & -- \\
LP-QAOA (two stages, separate phases) & 40 & 775,485 & $3.492\times10^{-2}$ & $1.008\times10^{8}$ & 10/10, 10/10 \\
LP-QAOA (two stages, combined phase) & 12 & 243,605 & $4.248\times10^{-3}$ & $2.635\times10^{8}$ & 10/10, 10/10 \\
Block-XY QAOA (separate phases) & 16 & 265,745 & $1.834\times10^{-3}$ & $6.668\times10^{8}$ & 10/10, 10/10 \\
Block-XY QAOA (combined phase) & 1 & 30,980 & $3.217\times10^{-6}$ & $4.434\times10^{10}$ & 10/10, 10/10 \\
\addlinespace[3pt]
No joint-feasibility stage & 4 & 92,085 & $8.348\times10^{-4}$ & $5.078\times10^{8}$ & 10/10, 10/10 \\
Warm-start block-XY & 32 & 514,445 & $5.324\times10^{-3}$ & $4.439\times10^{8}$ & 10/10, 10/10 \\
Uniform-state projector & 64 & 1,070,145 & $4.091\times10^{-4}$ & $1.204\times10^{10}$ & 10/10, 10/10 \\
\bottomrule
\end{tabular}%
}
\par\vspace{2pt}
\begin{minipage}{0.98\textwidth}\footnotesize
\textit{Notes.} Probabilities and total costs are geometric means over ten final-objective realizations after geometric-mean aggregation over four optimizer starts within each realization.  The last column counts realizations on which LP-QAOA has lower total cost and higher success probability than the listed construction.
\end{minipage}
\end{table}

The depth sweep showed the resource-adjusted depth profile for LP-QAOA (Table~\ref{tab:bcst_multilevel_fresh_depth}).  Success probability increased at every successive displayed depth from $p=8$ through $p=48$, while the additional replay cost placed the broad total-cost minimum at $p=8$--16.

\begin{table}[!htbp]
\centering
\caption{LP-QAOA depth dependence across ten realizations.}
\label{tab:bcst_multilevel_fresh_depth}
\small
\begin{tabular}{@{}rrrr@{}}
\toprule
$p$ & $P_{\mathrm{succ}}$ & Per-shot RU & $\mathrm{RTS}_{99}\times\RU$ \\
\midrule
0 & 0.002429 & 66,385 & $1.257\times10^{8}$ \\
1 & 0.004022 & 184,635 & $2.110\times10^{8}$ \\
2 & 0.008515 & 302,885 & $1.632\times10^{8}$ \\
4 & 0.019740 & 539,385 & $1.248\times10^{8}$ \\
6 & 0.035669 & 775,885 & $9.878\times10^{7}$ \\
8 & 0.053500 & 1,012,385 & $8.532\times10^{7}$ \\
10 & 0.063585 & 1,248,885 & $8.811\times10^{7}$ \\
11 & 0.079832 & 1,367,135 & $7.638\times10^{7}$ \\
12 & 0.084087 & 1,485,385 & $7.873\times10^{7}$ \\
14 & 0.087031 & 1,721,885 & $8.790\times10^{7}$ \\
16 & 0.088631 & 1,958,385 & $9.810\times10^{7}$ \\
24 & 0.092391 & 2,904,385 & $1.393\times10^{8}$ \\
32 & 0.096317 & 3,850,385 & $1.769\times10^{8}$ \\
40 & 0.101417 & 4,796,385 & $2.091\times10^{8}$ \\
48 & 0.103662 & 5,742,385 & $2.444\times10^{8}$ \\
\bottomrule
\end{tabular}
\end{table}

\FloatBarrier

The single-stage separate-phase block-XY circuit's $p=64$ success probability remained below its $p=16$ value and its total cost remained above that of LP-QAOA at $p=11$.  Complete depth curves and table-level outcomes for all circuit families are provided in Source Data.

\subsection{Standard QAOA diagnostic}
\label{app:bcst_standard_qaoa}

For the standard QAOA diagnostic, the full-space simulation was reduced to three sites and 18 binary variables.  The sites form a path, each site is constrained to have two occupied variables, and the quota is $(2,1,1,1,1,0)$, giving 12 jointly feasible configurations.  With $P$ the sum of the site-cardinality, quota and conflict penalties, the phase Hamiltonian and loss used
\begin{equation}
    H=(O_s+377P)/377.
    \label{eq:bcst_vanilla_penalty}
\end{equation}
The coefficient 377 places every infeasible configuration above a feasible witness for all ten final-objective realizations.  At $p=4$, the seven realizations identified in Source Data had geometric-mean success probability $1.503\times10^{-5}$, compared with $3.815\times10^{-6}$ at $p=0$, while total cost increased from $3.165\times10^9$ to $4.093\times10^9$.  The reduced calculation therefore shows shallow probability amplification without a resource advantage under the fixed budget.  Quantitative comparisons among the 30-variable circuits use the fixed-cardinality simulations reported above.

\FloatBarrier

\section{Additional two-level BCST numerical experiments}
\label{app:supplementary_numerics}

\subsection{Hybrid degree-four/degree-six BCST construction}
\label{app:bcst_details}

The BCST instances studied here have five blocks, with exactly two selected labels in each block.  Writing $x_{b,r}\in\{0,1\}$ for label $r$ in block $b$, the block-cardinality and conflict Hamiltonians are
\begin{align}
    C_{\mathrm{card}}(x)
    &=\sum_{b=1}^{5}\left(\sum_{r=1}^{K}x_{b,r}-2\right)^2,\\
    C_{\mathrm{conf}}(x)
    &=\sum_{\{u,v\}\in E_H}x_u x_v.
\end{align}
The block-XY mixer preserves the joint two-excitation subspace.  At $N=25$ ($K=5$), this product subspace has dimension $\binom{5}{2}^{5}=100{,}000$ and contains 120 conflict-feasible states.  At $N=30$ ($K=6$), it has dimension $\binom{6}{2}^{5}=759{,}375$ and contains 6,570 conflict-feasible states.

The final Hamiltonian combines signed four- and six-variable terms,
\begin{equation}
    C_{\mathrm{obj}}^{(s)}(x)
    =\sum_{T\in\mathcal T_{4}^{(s)}}w_T\prod_{u\in T}x_u
     +\sum_{T\in\mathcal T_{6}^{(s)}}w_T\prod_{u\in T}x_u,
    \label{eq:hybrid_bcst_objective}
\end{equation}
with fixed supports and weights for each instance $s$.  Each selected support $T$ receives the deterministic pseudorandom weight $w_T=(2z_T+1-2^{53})/2^{54}\in(-1/2,1/2)$, where $z_T$ is the 53-bit integer extracted from a SHA-256 hash of $(N,s,|T|,T)$.  The construction selects half of the eligible degree-four pool and sets the number of degree-six terms to one sixteenth of the selected degree-four count.  The $N=25$ instances therefore contain 225 degree-four and 14 degree-six terms; the $N=30$ instances contain 788 and 49, respectively.  Every reported instance has a unique conflict-feasible optimum, which is the primary target.  The best two and best eight conflict-feasible states are secondary sensitivity targets.

Here, XY--LP-QAOA denotes two-stage LP-QAOA with a block-XY mixer in stage~1.  The conflict-stage depth is $p_1=12$ at both sizes.  Stage~1 fixes the optimized preparation circuit for $\ket{\phi}$; stage~2 starts from $\ket{\phi}$, alternates the final-objective phase separator with its learned-projector mixer, and optimizes the stage-2 angles.  Stage-1-only measures $\ket{\phi}$ without stage-2 optimization.  In the matched comparison, warm-start block-XY refinement starts from the same learned stage-1 state but uses block-XY mixing with a conflict-plus-objective phase separator and loss.  The controlled ablation in \suppref{app:bcst_causal} separately changes only the mixer at fixed phase and loss.  Block-XY (combined phase) starts from the block-cardinality state and applies one angle $\gamma_\ell$ to the conflict-plus-objective Hamiltonian in each layer.  Block-XY (separate phases) uses the same initial state and mixer but assigns distinct angles $\gamma_{\mathrm{conf},\ell}$ and $\gamma_{\mathrm{obj},\ell}$ to the conflict and final-objective terms.

\subsection{Comparison sets and scope}
\label{app:bcst_cohorts}

Table~\ref{tab:bcst-cohort-registry} lists the instance sets and optimization settings used in each BCST analysis.

\begin{table}[!htbp]
\centering
\caption{BCST comparison sets.  Simultaneous perturbation stochastic approximation (SPSA) was used in the robustness analysis.  Counts refer to problem instances at each listed size.}
\label{tab:bcst-cohort-registry}
\small
\begin{tabularx}{\textwidth}{@{}>{\raggedright\arraybackslash}p{0.18\textwidth}>{\raggedright\arraybackslash}p{0.13\textwidth}>{\raggedright\arraybackslash}p{0.16\textwidth}>{\raggedright\arraybackslash}p{0.22\textwidth}>{\raggedright\arraybackslash}X@{}}
\toprule
Analysis & Size and count & Optimizer & Circuit depths & Role \\
\midrule
Robustness analysis & $N=30$, six & SPSA, five runs & selected and common depths & optimizer and depth sensitivity \\
Fixed-budget comparison & $N=25$, ten & Adam, 400 updates, three restarts & LP $p=12$; warm-start $p=1$ & common-budget comparison \\
BCST validation & 20 per size & Adam, 400 or 800 updates & method-specific depths & method comparison \\
Depth and resource analysis & $N=30$, two & Adam, 800 updates, three restarts & LP $p=1$--$128$ & depth--cost trade-off \\
Initial-state projector control & five per size & Adam, 400 updates & $p=24$ & initial state and projector reference \\
Mixer, reference and objective ablations & three per size & Adam, 800 updates, three restarts & $p=12/24$ & component attribution \\
Gradient diagnostics & $N=30$, six & random circuit angles & $p=26,29,32$ & finite-depth gradients \\
\bottomrule
\end{tabularx}
\end{table}

The robustness analysis used SPSA on two depth-tuning instances and six evaluation instances.  SPSA estimates a full gradient from two objective evaluations along a random perturbation direction.  The optimization protocol combined one run continued from the preceding depth with four independent starts and used 1,291 objective evaluations.  At the selected depths ($p=4$ for XY--LP-QAOA and $p=2$ for both block-XY parameterizations), the geometric-mean unique-optimum cost was lower by factors of 3.14 relative to the lower-cost block-XY parameterization and 1.98 relative to warm-start block-XY\@.  XY--LP-QAOA had lower cost on all six evaluation instances in both comparisons.  Best-two, best-eight and common-depth results are reported as sensitivity analyses.

\subsection{Fixed-budget comparison}
\label{app:bcst_sealed}

To test whether the cost ordering persists under a common fixed Adam budget, we used 400 updates, a learning rate of 0.035, and three independent restarts without continuation.  At $N=25$, XY--LP-QAOA had lower unique-optimum total cost, including stage-1 replay, than warm-start block-XY refinement on all ten instances.  The geometric mean of the per-instance warm-start block-XY/XY--LP-QAOA cost ratios was 6.11.

\begin{table}[t]
\centering
\caption{Fixed-budget comparison across ten $N=25$ instances using 400 Adam updates.  The ratio is the comparator's exact integer unique-optimum total cost divided by that of XY--LP-QAOA.}\label{tab:bcst-adam400-formal-comparisons}
\small
\begin{tabular}{@{}rllcc@{}}
\toprule
$N$ & Comparator & $p_{\rm LP}/p_{\rm comp.}$ & Geometric-mean cost ratio & Lower XY--LP-QAOA cost \\
\midrule
25 & Warm-start block-XY & $12/1$ & 6.11 & 10/10 \\ 
\bottomrule
\end{tabular}
\par\vspace{2pt}
\begin{minipage}{0.98\textwidth}\footnotesize
\textit{Notes.} For each instance, Adam was run for 400 updates with a learning rate of 0.035, three independent restarts, and no continuation.  Warm-start block-XY is initialized from the same stage-1 state as XY--LP-QAOA.
\end{minipage}
\end{table}

\FloatBarrier

\subsection{Matched method comparison}
\label{app:bcst_validation}

LP-QAOA also reduced total cost on the BCST problems with two constraint levels studied in this section.  The validation uses 20 random instances at each size.  Five instances per size used the selected LP-QAOA depths $p=16$ at $N=25$ and $p=96$ at $N=30$; the remaining fifteen used $p=12$ and $p=80$, respectively, while the comparator depths were unchanged.  Figure~\ref{fig:bcst_matched_scaling} reports all 20 instances.  Tables~\ref{tab:five-seed-matched-validation} and~\ref{tab:bcst-five-seed-target-geomeans} summarize five of them, including the best-two and best-eight sensitivity metrics.

\begin{figure}[!htbp]
\centering
\includegraphics[width=0.98\linewidth]{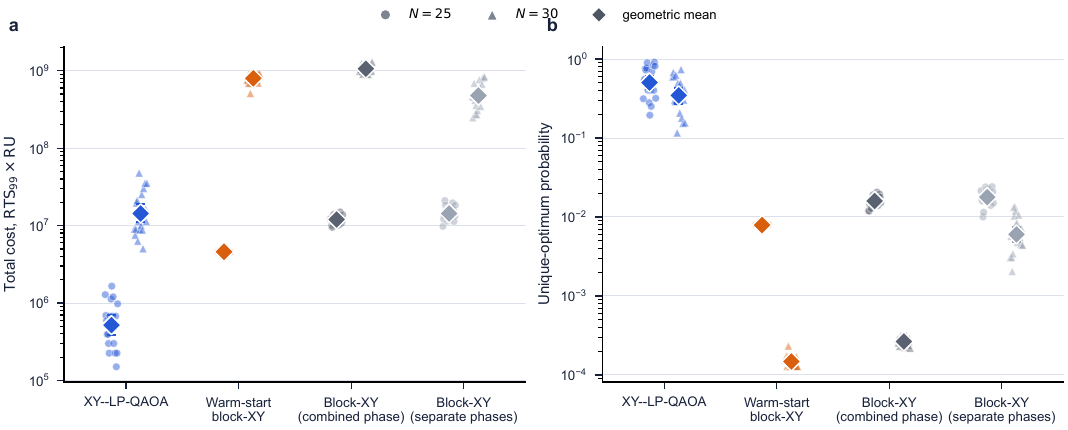}
\caption{\textbf{Two-level BCST comparisons.} \textbf{a,b}, Unique-optimum total cost and success probability for XY--LP-QAOA and matched block-XY comparators ($n=20$ instances per size). Small symbols show instances; diamonds and bars show geometric means and 95\% bootstrap intervals.}
\label{fig:bcst_matched_scaling}
\end{figure}

One block-XY baseline applied a single phase angle to the combined conflict-plus-objective Hamiltonian in each layer; the other assigned separate angles to the conflict and objective terms.  Including stage-1 replay and terminal verification, XY--LP-QAOA had lower unique-optimum total cost and higher unique-optimum probability than warm-start block-XY and both block-XY parameterizations on every instance at both sizes (Fig.~\ref{fig:bcst_matched_scaling}).

Warm-start block-XY retained the optimized stage-1 state but used local block-XY refinement with the conflict-penalized phase separator and loss.  Controlled mixer and reference-state substitutions under common depths and initial states are reported separately in \suppref{app:bcst_causal}.

\FloatBarrier

\begin{table*}[t]
\centering
\caption{Aggregate unique-optimum results over five BCST validation instances at each size.}
\label{tab:five-seed-matched-validation}
\small
\begin{tabularx}{\textwidth}{@{}>{\raggedright\arraybackslash}Xrrrrrr@{}}
\toprule
Method & $p$ & \shortstack{Geometric mean\\$P_{\mathrm{unique}}$} & \shortstack{Geometric mean\\feasible mass} & RU/shot & \shortstack{Geometric mean\\total cost} & \shortstack{Cost ratio\\comp./LP} \\
\midrule
\multicolumn{7}{@{}l}{\textit{$N=25$: matched method comparison}} \\
XY--LP-QAOA & 16 & 0.550 & 0.970 & 98,651 & $6.06\times10^{5}$ & 1.00 \\ 
Stage-1-only & -- & $7.82\times10^{-3}$ & 0.938 & 4,363 & $2.56\times10^{6}$ & 4.23 \\ 
Warm-start block-XY & 1 & $7.91\times10^{-3}$ & 0.949 & 7,896 & $4.58\times10^{6}$ & 7.57 \\ 
Block-XY (separate phases) & 12 & 0.0147 & 0.776 & 45,859 & $1.43\times10^{7}$ & 23.6 \\ 
Block-XY (combined phase) & 8 & 0.0132 & 0.868 & 31,727 & $1.10\times10^{7}$ & 18.2 \\ 
\addlinespace[2pt]
\multicolumn{7}{@{}l}{\textit{$N=25$: initial-state projector control}} \\
Initial-state projector\textsuperscript{\dag} & 24 & $2.47\times10^{-3}$ & 0.216 & 102,295 & $1.90\times10^{8}$ & 314 \\ 
\addlinespace[4pt]
\multicolumn{7}{@{}l}{\textit{$N=30$: matched method comparison}} \\
XY--LP-QAOA & 96\textsuperscript{*} & 0.461 & 0.983 & 1,505,215 & $1.20\times10^{7}$ & 1.00 \\ 
Stage-1-only & -- & $1.54\times10^{-4}$ & 0.972 & 13,375 & $4.00\times10^{8}$ & 33.2 \\ 
Warm-start block-XY & 1 & $1.56\times10^{-4}$ & 0.974 & 25,590 & $7.56\times10^{8}$ & 62.8 \\ 
Block-XY (separate phases) & 24 & $3.27\times10^{-3}$ & 0.970 & 305,275 & $4.30\times10^{8}$ & 35.7 \\ 
Block-XY (combined phase) & 4 & $2.65\times10^{-4}$ & 0.856 & 60,975 & $1.06\times10^{9}$ & 88.0 \\ 
\addlinespace[2pt]
\multicolumn{7}{@{}l}{\textit{$N=30$: initial-state projector control}} \\
Initial-state projector\textsuperscript{\dag} & 24 & $1.98\times10^{-4}$ & 0.799 & 325,315 & $7.57\times10^{9}$ & 629 \\ 
\bottomrule
\end{tabularx}
\par\vspace{2pt}
\begin{minipage}{0.98\textwidth}\footnotesize
\textit{Notes.} Probabilities, feasible masses and total costs are geometric means over five instances; per-shot RU is exact and invariant across instances.  Each instance-level total cost is the integer $\mathrm{RTS}_{99}\times\mathrm{RU}$ and includes frozen-stage replay.  Ratios exceed one when XY--LP-QAOA uses fewer logical resources.  Stage-1-only has no refinement stage.  Unrounded aggregates are supplied in Source Data.  \textsuperscript{*}\ The $N=30$ calculation uses $p=96$.  \textsuperscript{\dag}\ The initial-state projector uses the block-cardinality state as its initial state and projector reference.
\end{minipage}
\end{table*}
\begin{table*}[t]
\centering
\caption{Sensitivity of the BCST comparison to target-set size.  Unique-optimum results are reported in Supplementary Table~\ref{tab:five-seed-matched-validation}.}\label{tab:bcst-five-seed-target-geomeans}
\small
\begin{tabularx}{\textwidth}{@{}>{\raggedright\arraybackslash}Xlrrrr@{}}
\toprule
& & \multicolumn{2}{c}{Best two} & \multicolumn{2}{c}{Best eight} \\
\cmidrule(lr){3-4}\cmidrule(l){5-6}
Method & $p$/updates & $P_{\mathrm{succ}}$ & Total cost & $P_{\mathrm{succ}}$ & Total cost \\
\midrule
\multicolumn{6}{@{}l}{\textit{$N=25$}} \\
XY--LP-QAOA & 16/800 & 0.738 & $3.71\times10^{5}$ & 0.955 & $2.14\times10^{5}$ \\ 
Stage-1-only & -- & 0.0156 & $1.28\times10^{6}$ & 0.0625 & $3.14\times10^{5}$ \\ 
Warm-start block-XY & 1/400 & 0.0158 & $2.28\times10^{6}$ & 0.0633 & $5.61\times10^{5}$ \\ 
Block-XY (separate phases) & 12/400 & 0.0293 & $7.14\times10^{6}$ & 0.104 & $1.94\times10^{6}$ \\ 
Block-XY (combined phase) & 8/800 & 0.0267 & $5.41\times10^{6}$ & 0.0994 & $1.41\times10^{6}$ \\ 
\addlinespace[2pt]
Initial-state projector\textsuperscript{\dag} & 24/400 & $4.90\times10^{-3}$ & $9.60\times10^{7}$ & 0.0189 & $2.47\times10^{7}$ \\ 
\addlinespace[4pt]
\multicolumn{6}{@{}l}{\textit{$N=30$}} \\
XY--LP-QAOA & 96/800 & 0.720 & $5.75\times10^{6}$ & 0.959 & $2.62\times10^{6}$ \\ 
Stage-1-only & -- & $2.95\times10^{-4}$ & $2.09\times10^{8}$ & $1.23\times10^{-3}$ & $5.00\times10^{7}$ \\ 
Warm-start block-XY & 1/400 & $2.98\times10^{-4}$ & $3.95\times10^{8}$ & $1.25\times10^{-3}$ & $9.45\times10^{7}$ \\ 
Block-XY (separate phases) & 24/400 & $5.62\times10^{-3}$ & $2.49\times10^{8}$ & 0.0176 & $7.93\times10^{7}$ \\ 
Block-XY (combined phase) & 4/400 & $5.03\times10^{-4}$ & $5.58\times10^{8}$ & $1.89\times10^{-3}$ & $1.49\times10^{8}$ \\ 
\addlinespace[2pt]
Initial-state projector\textsuperscript{\dag} & 24/400 & $3.92\times10^{-4}$ & $3.82\times10^{9}$ & $1.51\times10^{-3}$ & $9.93\times10^{8}$ \\ 
\bottomrule
\end{tabularx}
\par\vspace{2pt}
\begin{minipage}{0.98\textwidth}\footnotesize
\textit{Notes.} Probabilities and total costs are geometric means over five instances; each instance-level cost uses integer $\mathrm{RTS}_{99}\times\mathrm{RU}$ and includes frozen-stage replay.  Unrounded aggregates are supplied in Source Data.  Stage-1-only has no refinement stage.  \textsuperscript{\dag}\ The initial-state projector uses the block-cardinality state as its initial state and projector reference.
\end{minipage}
\end{table*}

\FloatBarrier

Across the five instances in Tables~\ref{tab:five-seed-matched-validation} and~\ref{tab:bcst-five-seed-target-geomeans}, XY--LP-QAOA had lower unique-optimum total cost than every comparator at both sizes.  For each instance, the block-XY envelope used the lower observed cost of the combined- and separate-phase circuits.  Its geometric-mean cost ratio relative to XY--LP-QAOA was 17.3 at $N=25$ and 35.7 at $N=30$.  The initial-state projector control gave ratios of 314 and 629.

\subsection{Depth and resource analysis at \texorpdfstring{$N=30$}{N=30}}
\label{app:bcst_depth}

To determine whether additional refinement depth lowers total cost or only trades higher success probability for greater per-shot cost, we evaluated two $N=30$ tuning instances with 800 Adam updates and three independent restarts over the depth grid $p\in\{1,2,4,8,12,16,24,36,48,64,80,96,112,128\}$.  Unique-optimum probability reached its largest observed geometric mean at $p=128$, whereas unique-optimum total cost reached its lowest observed geometric mean at $p=80$.

\begin{figure}[!htbp]
\centering
\includegraphics[width=0.98\linewidth]{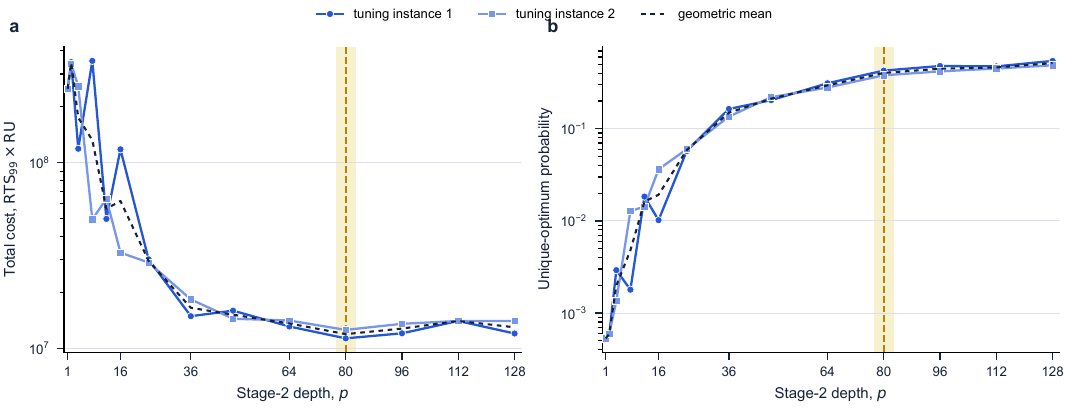}
\caption{\textbf{Depth dependence of LP-QAOA at $N=30$.} \textbf{a,b}, Unique-optimum total cost and success probability for two tuning instances; dotted curves show their geometric means. The highlighted depth minimizes total cost on the tested grid.}
\label{fig:bcst-n30-depth-resource}
\end{figure}

\begin{table*}[t]
\centering
\caption{Aggregate LP-QAOA depth and resource analysis at $N=30$ over two tuning instances.}
\label{tab:n30-lp-depth-resource}
\small
\begin{tabular}{@{}rrrr@{}}
\toprule
$p$ & \shortstack{Geometric mean\\$P_{\mathrm{unique}}$} & RU/shot & \shortstack{Geometric mean\\total cost} \\
\midrule
1 & $5.32\times10^{-4}$ & 28,915 & $2.50\times10^{8}$ \\ 
2 & $5.94\times10^{-4}$ & 44,455 & $3.45\times10^{8}$ \\ 
4 & $1.99\times10^{-3}$ & 75,535 & $1.75\times10^{8}$ \\ 
8 & $4.79\times10^{-3}$ & 137,695 & $1.32\times10^{8}$ \\ 
12 & 0.0162 & 199,855 & $5.64\times10^{7}$ \\ 
16 & 0.0192 & 262,015 & $6.21\times10^{7}$ \\ 
24 & 0.0587 & 386,335 & $2.95\times10^{7}$ \\ 
36 & 0.149 & 572,815 & $1.65\times10^{7}$ \\ 
48 & 0.212 & 759,295 & $1.52\times10^{7}$ \\ 
64 & 0.296 & 1,007,935 & $1.36\times10^{7}$ \\ 
\textbf{80}\textsuperscript{*} & 0.404 & 1,256,575 & $1.19\times10^{7}$ \\ 
96 & 0.449 & 1,505,215 & $1.28\times10^{7}$ \\ 
112 & 0.463 & 1,753,855 & $1.40\times10^{7}$ \\ 
128\textsuperscript{\dag} & 0.516 & 2,002,495 & $1.30\times10^{7}$ \\ 
\bottomrule
\end{tabular}
\par\vspace{2pt}
\begin{minipage}{0.98\textwidth}\footnotesize
\textit{Notes.} Each depth used 800 Adam updates.  Probability and unique-optimum total cost are geometric means over the two instances; per-shot RU follows Eq.~\eqref{eq:bcst_current_resource_counts}.  Instance-level probabilities and $\mathrm{RTS}_{99}$ values are supplied in Source Data.  \textsuperscript{*}~$p=80$ has the lowest displayed total cost on the tested grid.  \textsuperscript{\dag}~$p=128$ is the maximum tested depth and has the largest observed probability.
\end{minipage}
\end{table*}

\FloatBarrier

\subsection{Mixer, reference-state and objective ablations}
\label{app:bcst_causal}

The controlled ablation used three prespecified instances at each size and one learned stage-1 reference state $\ket{\phi}$ per size, shared across methods and instances.  Each optimized method used 800 Adam updates and three independent restarts at $p=12$ for $N=25$ and $p=24$ for $N=30$.  The learned-projector and warm-start block-XY methods under $H_{\mathrm{obj}}$ differ only in the stage-2 mixer; both start from $\ket{\phi}$ and use the normalized final-objective Hamiltonian $H_{\mathrm{obj}}$ for the phase separator and loss.  The corresponding methods under $H_{\mathrm{pen}}$ make the same mixer comparison.  The fixed-reference projector under $H_{\mathrm{pen}}$ instead replaces $\ket{\phi}$ in the projector mixer with the known block-cardinality state $\ket{s}$, whereas Stage-1-only applies no refinement.  Full trajectories and restart diagnostics are provided in Source Data.

Writing $C_{\mathrm{obj}}$ for the raw objective diagonal, let $c_{\min}$ and $c_{\max}$ denote its extrema over the complete block-cardinality basis.  The normalized final-objective Hamiltonian is $H_{\mathrm{obj}}=(C_{\mathrm{obj}}-c_{\min}\Id)/(c_{\max}-c_{\min})$.  For the conflict-count diagonal $C_{\mathrm{conf}}$, we set $H_{\mathrm{conf}}=C_{\mathrm{conf}}/10$ and $G=H_{\mathrm{obj}}+11H_{\mathrm{conf}}$, then define
\begin{equation*}
H_{\mathrm{pen}}=\frac{G-G_{\min}\Id}{G_{\max}-G_{\min}},
\end{equation*}
where $G_{\min}$ and $G_{\max}$ are likewise taken over this basis.  Before $G$ is normalized to form $H_{\mathrm{pen}}$, $H_{\mathrm{obj}}$ ranges from 0 to 1, whereas a single conflict adds 1.1 through $11H_{\mathrm{conf}}$.  Hence every conflict-free configuration lies below every conflicting configuration.  The $H_{\mathrm{obj}}$-to-$H_{\mathrm{pen}}$ comparison changes the phase separator and optimized loss together.

With the notation of \suppref{app:resource}, the per-shot costs for these methods are
\begin{align}
\RU_{\mathrm{LP},H_{\mathrm{obj}}}(p)&=P_1+p(O+2P_1+S)+O,\\
\RU_{\mathrm{LP},H_{\mathrm{pen}}}(p)&=P_1+p(C+O+2P_1+S)+O,\\
\RU_{\mathrm{warm\text{-}XY},H_{\mathrm{obj}}}(p)&=P_1+p(O+X)+O,\\
\RU_{\mathrm{warm\text{-}XY},H_{\mathrm{pen}}}(p)&=P_1+p(C+O+X)+O,\\
\RU_{\mathrm{fixed\text{-}reference},H_{\mathrm{pen}}}(p)&=P_1+p(C+O+2P_{\mathrm{blk}}+S)+O,\\
\RU_{\mathrm{initial\text{-}state}}(p)&=P_{\mathrm{blk}}+p(C+O+2P_{\mathrm{blk}}+S)+O.
\end{align}

At $N=25$ and $p=12$, these expressions give
\begin{equation*}
\begin{aligned}
\RU_{\mathrm{LP},H_{\mathrm{obj}}}&=75{,}079,
&\RU_{\mathrm{warm\text{-}XY},H_{\mathrm{obj}}}&=46{,}459,\\
\RU_{\mathrm{LP},H_{\mathrm{pen}}}&=75{,}379,
&\RU_{\mathrm{warm\text{-}XY},H_{\mathrm{pen}}}&=46{,}759,\\
\RU_{\mathrm{fixed\text{-}reference},H_{\mathrm{pen}}}&=53{,}779.
\end{aligned}
\end{equation*}
At $N=30$ and $p=24$, the corresponding values are
\begin{equation*}
\begin{aligned}
\RU_{\mathrm{LP},H_{\mathrm{obj}}}&=386{,}335,
&\RU_{\mathrm{warm\text{-}XY},H_{\mathrm{obj}}}&=305{,}815,\\
\RU_{\mathrm{LP},H_{\mathrm{pen}}}&=387{,}055,
&\RU_{\mathrm{warm\text{-}XY},H_{\mathrm{pen}}}&=306{,}535,\\
\RU_{\mathrm{fixed\text{-}reference},H_{\mathrm{pen}}}&=326{,}575.
\end{aligned}
\end{equation*}
The initial-state projector controls use
\begin{equation*}
\RU_{\mathrm{initial\text{-}state}}(p)=
\begin{cases}
3{,}463+4{,}118p,&N=25,\\
12{,}115+13{,}050p,&N=30,
\end{cases}
\end{equation*}
which gives 102,295 and 325,315 RU at $p=24$.

Unique-optimum total cost, $\mathrm{RTS}_{99}\times\mathrm{RU}$, gives the primary comparison.  Under both $H_{\mathrm{obj}}$ and $H_{\mathrm{pen}}$, learned-projector refinement had lower total cost and higher unique-optimum probability than warm-start block-XY on every tested instance at both sizes.  Under $H_{\mathrm{pen}}$, the learned reference likewise outperformed the fixed block-cardinality reference on both outcomes.  Absolute outcomes are shown in Fig.~\ref{fig:bcst-causal-ablation}; paired probability contrasts are reported in Table~\ref{tab:causal-unique-contrasts}, and instance-level values are supplied in Source Data.

\begin{figure}[!htbp]
\centering
\includegraphics[width=0.98\linewidth]{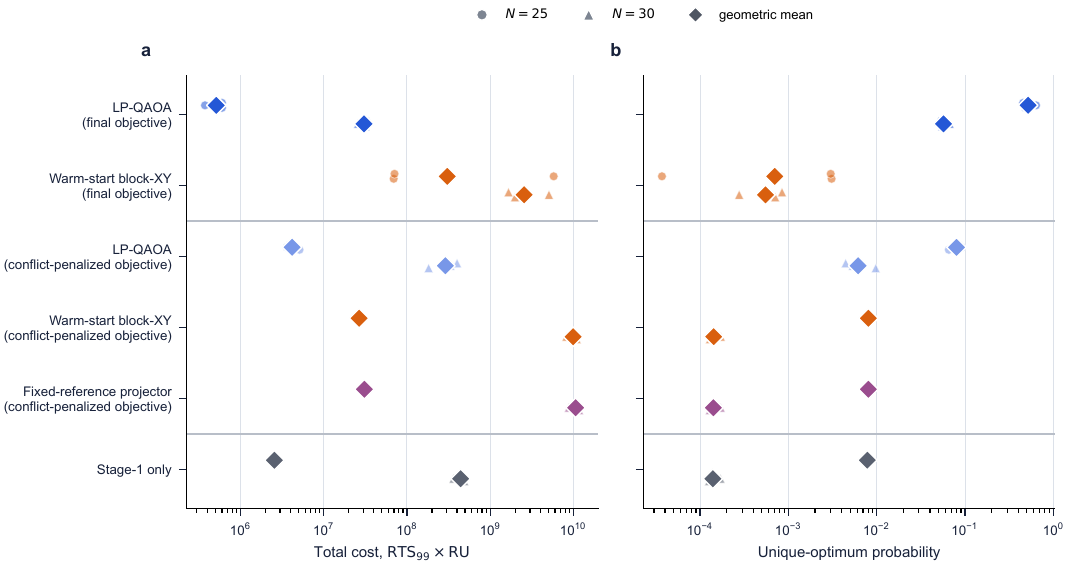}
\caption{\textbf{BCST mixer and reference-state ablations.} \textbf{a,b}, Unique-optimum total cost and success probability at $N=25$ and $N=30$ ($n=3$ instances per size). Small symbols show instances; diamonds show geometric means.}
\label{fig:bcst-causal-ablation}
\end{figure}

\begin{table*}[t]
\caption{Circuit definitions and per-shot resources for the controlled BCST ablations.}
\label{tab:causal-arm-registry}
\centering
\small
\begin{tabularx}{\textwidth}{@{}>{\raggedright\arraybackslash}Xlllrr@{}}
\toprule
Method & Mixer & Phase/loss & Projector ref. & $p$ & RU \\
\midrule
\multicolumn{6}{@{}l}{\textit{$N=25$}} \\
Learned projector ($H_{\mathrm{obj}}$) & $\Id-\lvert\phi\rangle\!\langle\phi\rvert$ & $H_{\mathrm{obj}}$ & $\phi$ & 12 & 75{,}079 \\
Warm-start block-XY ($H_{\mathrm{obj}}$) & $M_{\mathrm{XY}}$ & $H_{\mathrm{obj}}$ & -- & 12 & 46{,}459 \\
Learned projector ($H_{\mathrm{pen}}$) & $\Id-\lvert\phi\rangle\!\langle\phi\rvert$ & $H_{\mathrm{pen}}$ & $\phi$ & 12 & 75{,}379 \\
Warm-start block-XY ($H_{\mathrm{pen}}$) & $M_{\mathrm{XY}}$ & $H_{\mathrm{pen}}$ & -- & 12 & 46{,}759 \\
Fixed-reference projector ($H_{\mathrm{pen}}$) & $\Id-\lvert s\rangle\!\langle s\rvert$ & $H_{\mathrm{pen}}$ & $s$ & 12 & 53{,}779 \\
Stage-1-only & -- & -- & -- & -- & 4{,}363 \\
\midrule
\multicolumn{6}{@{}l}{\textit{$N=30$}} \\
Learned projector ($H_{\mathrm{obj}}$) & $\Id-\lvert\phi\rangle\!\langle\phi\rvert$ & $H_{\mathrm{obj}}$ & $\phi$ & 24 & 386{,}335 \\
Warm-start block-XY ($H_{\mathrm{obj}}$) & $M_{\mathrm{XY}}$ & $H_{\mathrm{obj}}$ & -- & 24 & 305{,}815 \\
Learned projector ($H_{\mathrm{pen}}$) & $\Id-\lvert\phi\rangle\!\langle\phi\rvert$ & $H_{\mathrm{pen}}$ & $\phi$ & 24 & 387{,}055 \\
Warm-start block-XY ($H_{\mathrm{pen}}$) & $M_{\mathrm{XY}}$ & $H_{\mathrm{pen}}$ & -- & 24 & 306{,}535 \\
Fixed-reference projector ($H_{\mathrm{pen}}$) & $\Id-\lvert s\rangle\!\langle s\rvert$ & $H_{\mathrm{pen}}$ & $s$ & 24 & 326{,}575 \\
Stage-1-only & -- & -- & -- & -- & 13{,}375 \\
\bottomrule
\end{tabularx}
\par\vspace{2pt}\footnotesize\textit{Notes.} All methods start from the same learned stage-1 state \(\lvert\phi\rangle\) within each size.  Each optimized method used 800 Adam updates and three independent restarts at the listed common depth; Stage-1-only has no refinement stage.  \(H_{\mathrm{obj}}\) and \(H_{\mathrm{pen}}\) specify both the phase separator and loss.  Outcomes are reported in Fig.~\ref{fig:bcst-causal-ablation}.
\end{table*}
\begin{table*}[t]
\caption{Aggregate unique-optimum probability contrasts for the controlled BCST ablations.}
\label{tab:causal-unique-contrasts}
\centering
\small
\begin{tabularx}{\textwidth}{@{}r>{\raggedright\arraybackslash}Xrrr@{}}
\toprule
$N$ & Contrast (left minus right) & Mean & Range & Positive instances \\
\midrule
25 & Learned projector ($H_{\mathrm{obj}}$) $-$ warm-start block-XY ($H_{\mathrm{obj}}$) & $0.522$ & $0.452$--$0.637$ & 3/3 \\
25 & Learned projector ($H_{\mathrm{pen}}$) $-$ warm-start block-XY ($H_{\mathrm{pen}}$) & $0.0728$ & $0.0575$--$0.0840$ & 3/3 \\
25 & Learned projector ($H_{\mathrm{pen}}$) $-$ fixed-reference projector ($H_{\mathrm{pen}}$) & $0.0729$ & $0.0575$--$0.0841$ & 3/3 \\
\midrule
30 & Learned projector ($H_{\mathrm{obj}}$) $-$ warm-start block-XY ($H_{\mathrm{obj}}$) & $0.0569$ & $0.0494$--$0.0667$ & 3/3 \\
30 & Learned projector ($H_{\mathrm{pen}}$) $-$ warm-start block-XY ($H_{\mathrm{pen}}$) & $0.00637$ & $0.00432$--$0.00962$ & 3/3 \\
30 & Learned projector ($H_{\mathrm{pen}}$) $-$ fixed-reference projector ($H_{\mathrm{pen}}$) & $0.00637$ & $0.00432$--$0.00962$ & 3/3 \\
\bottomrule
\end{tabularx}
\par\vspace{2pt}\footnotesize\textit{Notes.} Each entry is \(P_{\mathrm{unique}}(\mathrm{left})-P_{\mathrm{unique}}(\mathrm{right})\); positive values favour the named method on the left.  Means and ranges summarize three instances at each size; sizes are not pooled.  Instance-level contrasts are supplied in Source Data.
\end{table*}

\FloatBarrier

\subsection{Finite-depth trainability under learned-projector mixing}
\label{app:gradient_diagnostics}

To evaluate finite-depth trainability under learned-projector and local block-XY mixing, we computed gradients at randomly sampled angles on six $N=30$ BCST instances.  Across depths 26, 29 and 32, XY--LP-QAOA had normalized training-loss gradient magnitudes that were 73.7--89.5-fold larger than the three block-XY baselines.  The corresponding final-objective gradient magnitudes were 68.9--89.0-fold larger (Table~\ref{tab:app_gradient_diagnostics}).  The baselines were block-XY (combined phase), block-XY (separate phases) and warm-start block-XY refinement.  Training-loss descent was more closely aligned with increasing unique-optimum probability in seven of nine comparisons across the three baselines and three depths.  The warm-start block-XY comparisons control for the shared stage-1 initialization.  These measurements quantify local sensitivity to circuit parameters.

For each of six instances, four methods and three depths, we sampled 12 angle vectors, giving 864 evaluated angle settings.  Angles were sampled uniformly from $[-\pi,\pi]$, with phase and mixer angles shared across methods and additional constraint-phase angles sampled independently for separate-phase block-XY.  At every setting, we calculated gradients of the training loss, final objective and unique-optimum log probability, giving 2,592 gradient vectors in total.  For XY--LP-QAOA, gradients were taken with respect to the stage-2 objective and learned-projector-mixer angles while the learned structural state was held fixed.

For $\mathcal L_H(\theta)=\langle H\rangle_\theta$, we measured normalized gradient magnitude as
\begin{equation}
    g_{\mathrm{RMS}}^{(H)}(\theta)
    =\frac{1}{\sigma_H}\sqrt{\frac{1}{K}\sum_{a=1}^{K}
    \left(\frac{\partial\mathcal L_H}{\partial\theta_a}\right)^2},
    \label{eq:gradient_rms_normalized}
\end{equation}
where $K$ is the number of variable circuit angles and $\sigma_H$ is the standard deviation of the diagonal energies over the complete block-cardinality subspace.  RMS denotes the root mean square across parameter derivatives.  For the training loss, $H$ is each method's training Hamiltonian; for final-objective gradients, every method uses the same final Hamiltonian for that instance.  Each instance was summarized by its median over the 12 angle vectors.  Gradient ratios were then aggregated as geometric means over the six paired instances.

The alignment measure is the cosine between training-loss descent and the gradient of unique-optimum log probability.  For $P_{\mathrm{succ}}>0$, the identity
\begin{equation}
    \nabla_\theta\log P_{\mathrm{succ}}
    =\frac{\nabla_\theta P_{\mathrm{succ}}}{P_{\mathrm{succ}}}
\end{equation}
shows that using probability or log probability gives the same direction and hence the same cosine.  Figure~\ref{fig:gradient_trainability} therefore describes this quantity as alignment with increasing unique-optimum probability.  Open symbols show the six instance medians; filled symbols and bars show their median and interquartile range.

\begin{table}[!htbp]
\centering
\caption{Learned-projector mixing yields larger, more target-aligned gradients at $N=30$.  Gradient-RMS ratios divide XY--LP-QAOA by the listed comparator; values above one favour XY--LP-QAOA\@.  Gradient--target alignment is the difference between the XY--LP-QAOA and comparator cosine values for training-loss descent and increasing unique-optimum log probability; positive values indicate closer target alignment.  Entries aggregate six paired instance medians.}
\label{tab:app_gradient_diagnostics}
\small
\begin{tabular}{rlrrr}
\toprule
Depth & Comparator & \shortstack{Training-loss\\RMS ratio} & \shortstack{Final-objective\\RMS ratio} & \shortstack{Gradient--target\\alignment difference} \\
\midrule
26 & Block-XY (combined phase) & $77.60$ & $77.10$ & $+0.073$ \\
26 & Block-XY (separate phases) & $87.95$ & $89.03$ & $+0.099$ \\
26 & Warm-start block-XY & $86.76$ & $83.89$ & $+0.040$ \\
29 & Block-XY (combined phase) & $79.72$ & $77.41$ & $+0.126$ \\
29 & Block-XY (separate phases) & $89.48$ & $86.70$ & $+0.123$ \\
29 & Warm-start block-XY & $85.39$ & $82.87$ & $+0.147$ \\
32 & Block-XY (combined phase) & $73.67$ & $68.86$ & $-0.046$ \\
32 & Block-XY (separate phases) & $86.03$ & $82.99$ & $+0.005$ \\
32 & Warm-start block-XY & $74.34$ & $77.71$ & $-0.091$ \\
\bottomrule
\end{tabular}
\end{table}

Across all three depths and comparators, the learned-projector circuit had larger normalized gradients; its gradient--target alignment was higher in seven of the nine comparisons.

\FloatBarrier

\section{Soft coarse-to-fine optimization on SBM instances}
\label{app:soft_multiscale}

\subsection{Graph ensemble and hierarchy}
\label{app:soft_sbm_csp}

We generated ten $N=24$ two-community stochastic block model (SBM) graphs with $(p_{\mathrm{in}},p_{\mathrm{out}})=(0.80,0.36)$ \cite{holland1983stochastic}.  The planted assignment is constant on six contiguous groups of four vertices and alternates between the two communities,
\[
    111100001111000011110000.
\]
For each unordered vertex pair, an edge was drawn independently with probability $p_{\mathrm{in}}$ when the planted labels agree and $p_{\mathrm{out}}$ otherwise.  Ten independently sampled graph realizations were retained.

For adjacency matrix $A$, the final Hamiltonian is
\begin{equation}
    C_{\mathrm{fine}}(z)
    =-\sum_{1\le i<j\le 24}(A_{ij}-1/2)z_i z_j,
    \qquad z_i\in\{-1,+1\}.
    \label{eq:app_sbm_fine}
\end{equation}
The exact target set
\[
    \calO_{\mathrm{SBM}}=\argmin_z C_{\mathrm{fine}}(z)
\]
was enumerated over the full computational basis.  Each graph has two ground states related by global spin flip.

The coarse Hamiltonian uses the planted four-vertex grouping as a hierarchy prior.  Let $G_a$ denote group $a$ and
\[
    \bar J_{ab}=\frac{1}{|G_a||G_b|}
    \sum_{i\in G_a,j\in G_b}(A_{ij}-1/2)
\]
be the mean coupling between groups.  The implemented coarse objective is
\begin{equation}
    C_{\mathrm{coarse}}(z)
    =-\sum_a\sum_{i<j\in G_a}z_i z_j
    -0.25\sum_{a<b}\bar J_{ab}
    \sum_{i\in G_a,j\in G_b}z_i z_j.
    \label{eq:app_sbm_coarse}
\end{equation}
This objective favours agreement within each four-vertex group and represents inter-group relations through averaged couplings.  We call a graph \emph{aligned} when every fine ground state is also a coarse ground state.  Nine graphs were aligned.  In the remaining graph, 628 computational-basis states have coarse energy strictly below the fine-ground-state energy.  This graph serves as the coarse--fine mismatch control.

\subsection{Circuits and fixed-budget optimization}
\label{app:soft_sbm_methods}

Stage 1 of LP-QAOA uses $p_1=4$ transverse-field mixer layers and coarse phase separators,
\[
    \ket{\phi_1(\theta_1)}
    =\prod_{\ell=p_1}^{1}
    e^{-i\beta_{1,\ell}\sum_iX_i}
    e^{-i\gamma_{1,\ell}C_{\mathrm{coarse}}}
    \ket{+}^{\otimes24},
\]
and minimizes the expected coarse energy.  After freezing the selected stage-1 circuit, stage 2 applies one fine phase separator and one learned-projector mixer,
\[
    \ket{\psi_2(\theta_2)}
    =e^{-i\beta_{2,1}R_{\phi_1}}
    e^{-i\gamma_{2,1}C_{\mathrm{fine}}}
    \ket{\phi_1},
    \qquad
    R_{\phi_1}=\Id-\ket{\phi_1}\bra{\phi_1},
\]
and minimizes the expected fine energy.

The comparisons include the selected stage-1 circuit alone (LP stage-1 objective QAOA), a warm-start circuit that follows the same frozen preparation with one fine phase separator and one transverse-field mixer, and standard QAOA on $C_{\mathrm{fine}}$ at $p=4,6,8$.  Every method used Adam for 320 parameter updates from four independent initializations.  Learning rates were $0.04$ for stage 1, stage 2 and standard QAOA, and $0.02$ for warm-start QAOA\@.  One initialization used a smooth angle schedule and the other three sampled each angle uniformly from $[-\pi,\pi]$.  For each method and graph, the checkpoint and initialization with the lowest expected training energy were selected before ground-state probability was evaluated.  All state-vector calculations used $2^{24}$ complex amplitudes with 32-bit real and imaginary components.  These comparisons use the same optimization budget for every graph at a given method and depth.

\subsection{Aligned-graph comparisons}
\label{app:soft_sbm_aligned}

Resource units follow Eqs.~\eqref{eq:soft_ru} and \eqref{eq:sbm_continuous_cost}.  At $N=24$, the fine and coarse phase separators each contain $\binom{24}{2}=276$ two-body terms, the transverse-field mixer costs 24 RU and the selective phase costs $24^2=576$ RU\@.  The frozen coarse preparation therefore costs
\begin{equation*}
    \RU(U_1)=4(276+24)=1{,}200.
\end{equation*}
The per-preparation costs used in the comparison are
\begin{align*}
    \RU_{\mathrm{direct}}(p)&=p(276+24)+276,\\
    \RU_{\mathrm{stage\text{-}1}}(4)&=\RU_{\mathrm{standard}}(4)=1{,}476,\\
    \RU_{\mathrm{standard}}(6)&=2{,}076,
    &\RU_{\mathrm{standard}}(8)&=2{,}676,\\
    \RU_{\mathrm{warm}}(4,1)&=1{,}200+(276+24)+276=1{,}776,\\
    \RU_{\mathrm{LP}}(4,1)&=1{,}200+(276+2\times1{,}200+576)+276=4{,}728.
\end{align*}
Table~\ref{tab:app_soft_sbm} reports medians over the nine aligned graphs and paired LP-QAOA win counts.  The standard QAOA probability envelope takes the largest probability among $p=4,6,8$ separately for each graph; its cost envelope takes the smallest continuous cost.

\begin{table}[!htbp]
\centering
\caption{SBM outcomes on the nine coarse--fine aligned graphs at $(p_{\mathrm{in}},p_{\mathrm{out}})=(0.80,0.36)$.}
\label{tab:app_soft_sbm}
\small
\resizebox{\textwidth}{!}{%
\begin{tabular}{llrrrrr}
\toprule
Method & Depth & RU & Median $P_{\mathrm{succ}}$ & LP probability wins & Median $\mathcal C_{99}^{\mathrm{cont}}$ & LP cost wins \\
\midrule
LP-QAOA & $(4,1)$ & 4,728 & 0.7771 & -- & 14,505 & -- \\
LP stage-1 objective QAOA & $4$ & 1,476 & 0.2932 & 9/9 & 19,592 & 9/9 \\
Warm-start QAOA & $(4,1)$ & 1,776 & 0.2902 & 9/9 & 23,856 & 9/9 \\
Standard QAOA & $4$ & 1,476 & 0.1655 & 9/9 & 37,577 & 9/9 \\
Standard QAOA & $6$ & 2,076 & 0.3385 & 9/9 & 23,137 & 6/9 \\
Standard QAOA & $8$ & 2,676 & 0.4338 & 9/9 & 21,666 & 7/9 \\
Standard QAOA envelope & $4,6,8$ & -- & 0.4338 & 9/9 & 20,379 & 6/9 \\
\bottomrule
\end{tabular}}
\end{table}

LP-QAOA assigned more probability to the ground-state pair than every comparator on every graph in the main comparison.  The continuous-cost comparison also accounts for the larger circuit required to replay and uncompute the stage-1 preparation.  LP-QAOA had lower cost than the standard QAOA depth envelope on six of these nine graphs.

\subsection{Coarse--fine mismatch control}
\label{app:soft_sbm_mismatch}

In the coarse--fine mismatch control, the fine optimum differs from the planted assignment by two spin flips, while the coarse Hamiltonian favours the planted basin.  LP-QAOA assigned ground-state probability $2.021\times10^{-6}$ at continuous cost $1.078\times10^{10}$; standard QAOA at $p=4,6,8$ reached probabilities of $0.05444$--$0.1891$ at costs of $5.50\times10^4$--$1.21\times10^5$.  LP-QAOA reached the lowest expected fine energy, $-41.868$ compared with $-41.630$ for standard QAOA at $p=8$, while its output distribution remained concentrated in the coarse-favoured low-energy basin.  This control identifies coarse--fine ground-state alignment as a condition for amplification in this benchmark.  Complete outcomes, including the LP stage-1 objective and warm-start QAOA comparisons, are provided in Source Data.

\FloatBarrier

\ifdefined\supplementaryonly
\clearpage
\printbibliography[title={Supplementary References}]
\fi
\fi
\end{document}